\documentclass[11pt,reqno,makeidx]{amsart}
\usepackage[active]{srcltx}
\usepackage{amssymb,color}
\usepackage{float}
\restylefloat{table}

\usepackage{ragged2e}
\usepackage{array}
\usepackage[active]{srcltx}
\usepackage{epsfig}
\usepackage{hyperref}
\usepackage{amsmath}
\usepackage{amssymb}
\usepackage{hyperref}
\usepackage{enumerate}

\newtheorem{Proposition}{Proposition}
  \newtheorem{Remark}[Proposition]{Remark}
  \newtheorem{Corollary}[Proposition]{Corollary}
  \newtheorem{Lemma}[Proposition]{Lemma}
  
  \newtheorem{Theorem}[Proposition]{Theorem}
 
 \newtheorem{Definition}[Proposition]{Definition}

\def\CC{\mathbb{C}}

 \def\RR{\mathbb{R}}
 \def\NN{\mathbb{N}}

\def\Re{\mathrm{Re\,}}

\def\hb{\hbar}

\def\be{\begin{equation}}
\def\bel{\begin{equation}\label}
\def\ee{\end{equation}}
    \def\la{\langle}
\def\ra{\rangle}
\def\lam{\lambda}
\def\bd{\begin{Definition}}
\def\ed{\end{Definition}}

\def\bp{\begin{Proposition}}
\def\bpl{\begin{Proposition}\label}
\def\ep{\end{Proposition}}
\def\bl{\begin{Lemma}}
\def\el{\end{Lemma}}
\def\bt{\begin{Theorem}}
\def\et{\end{Theorem}}

\def\lam{\lambda}

\def\Bx{\ $\Box$}
 \def\la{\langle}
\def\ra{\rangle}
\def\calS{\mathcal{S}}

 \def\nv{ Y }

\def\al{\alpha}
\def\a{\alpha}
\def\sa{\sqrt{\alpha}}
\def\om{\omega}

\def\Le{\leqslant}
\def\Ge{\geqslant}
\def\R{\mathbb{R}}

\begin{document}

\title[Weber equation as a normal form with applications]{The Weber equation as a normal form with applications to top of the barrier scattering}
\author{Rodica D. Costin, Hyejin Park, Wilhelm Schlag}
\address{Department of Mathematics\\The Ohio State University\\Columbus, OH 43210}
  \email{costin.10@osu.edu}
\address{Department of Mathematics\\The Ohio State University\\Columbus, OH 43210}
\email{ piaohuizhen@gmail.com }
 \address{Department of Mathematics\\5734 South University Avenue\\The University of Chicago\\ Chicago IL 60637
} 
\email{schlag@math.uchicago.edu}

\thanks{The third author was partially supported by the NSF, DMS1160817. The authors would like to thank Ovidiu Costin for very helpful discussions and suggestions. RDC is grateful for the hospitality of the University of Chicago, and WS is grateful for the hospitality of the Ohio State University while collaborating on this project.}

\maketitle

\begin{abstract}

In the paper we revisit the basic problem of tunneling near a nondegenerate global maximum of a potential on the line.
We reduce the semiclassical Schr\"odinger equation to a Weber normal form by means of the Liouville-Green transform.
We show that the diffeomorphism which effects this stretching of the independent variable lies in the same regularity class as the potential
(analytic or infinitely differentiable) with respect to both variables, i.e., space and energy.
We then apply the Weber normal form to the scattering problem for energies near the potential maximum. In particular we obtain
a representation of the scattering matrix which is accurate up to multiplicative factors of the form 1 + o(1). 

\end{abstract}

\section{Introduction}

This paper deals with fine properties of the resolvent and the spectral measure of 
  Schr\"odinger equations
\bel{eqf}
-\hb^2\psi''(\xi)+V(\xi)\psi(\xi)=E\psi(\xi)
\ee
on the line where the potential satisfies the following properties: 
\begin{itemize}
\item  $V\in L^1(\RR)$
\item   $V$ is smooth, $C^\infty$ or analytic (for short $V\in C^v(\RR)$ with $v=\infty$ or $v=\omega$),
\item   $V(\xi)$ has a unique absolute maximum, say at $\xi=0$, where $V(\xi)=1-\xi^2+O(\xi^3)$.
\end{itemize} 
 We consider energies close to the {\em top of the potential barrier} $\max V=1$, and we wish to
 obtain accurate representations of the resolvent near this top energy uniformly in small $\hbar$. 
 We cannot do justice to the vast literature devoted to the equation~\eqref{eqf} and its higher-dimensional
 analogues.  For example, see \cite{ ABR, Bleher, BFSZ_REV, BFSZ, BFSZ2, BCD1, BCD2, CdVP1, CdVP2, Dyatlov, GerGri,  HS, HS1, Marz, Ramond, Sjo} 
and references cited there. For the most part, these  papers   deal with the asymptotic law of {\em resonances} in the limit $\hbar\to0$ 
in this setting (with~\cite{Dyatlov} being devoted to resonances generated by Kerr-deSitter black holes), as given by the {\em Bohr-Sommerfeld} 
 quantization condition.  This reduces to studying the asymptotic behavior as $\hbar \to0$ of solutions to the equation $Pu=Eu$ (where $P$ is the left-hand
 side of~\eqref{eqf}, for example) with the spectral parameter $E$ being $O(\hbar)$ close to the potential maximum. The underlying mechanism
 is a tunneling effect near the potential barrier. 
 
 Technically speaking, the methods employed vary, but  involve the analysis of the classical Hamiltonian flow
 near a hyperbolic fixed point, such as $(0,0)$ for the classical symbol $p_0(x,\xi)=\xi^2 + V(x)$, microlocal analysis of the resolvent operator, and complex WKB techniques. 
The latter requires analytic potentials. 
The interest in complex resonances resides inter alia with the fact that they enter into a description of the Schr\"odinger time evolution for long times, but
not infinite times (the finite threshold being the so-called ``Ehrefest time"). 

Here our focus is precisely on dispersive estimates for the Schr\"odinger evolution for {\em all times}, which requires very accurate control of the spectral measure associated with~\eqref{eqf}. 
In the context of the wave equation on a Schwarzschild black hole such dispersive estimates were obtained in~\cite{DSS}.  By means of an angular momentum decomposition~\cite{DSS} reduces
matters to an equation of the form~\eqref{eqf} for each fixed angular momentum. Two issues arise by doing so:  (i) the infinite time control of the evolution for fixed angular momentum (ii) the summation
problem, i.e., being able to sum up the resulting bounds over all angular momenta. 

As for (i), the main energy is as usual $0$, and precise control of the spectral measure is needed both in terms of small energies and the semi-classical parameter $\hbar=\ell^{-1}$ where $\ell$
is the angular momentum. Reference~\cite{CDST} develops this aspect of the theory. 

As for (ii), the summation problem hinges crucially on the fact that the potential $V$ has a nondegenerate maximum. Indeed, if, say $V$ had a trapping well (a local minimum at $x=0$), then
the constants in the estimates obtain for (i) would depend exponentially on some power of $\ell$ and therefore summation, if possible, requires a different approach.  However, the presence of a global nondegenerate 
maximum guarantees that the losses are only
in terms of some fixed power of~$\ell$ and therefore the summation can be carried out.  In~\cite{DSS} the scattering theory near the top of the potential barrier is based on Mourre theory and the propagation 
estimates of~\cite{HSS}. The idea here is that while the maximum energy corresponds to a classically trapping point $x=0, \xi=0$ in phase space, due to the uncertainty principle the basic
Mourre commutator remains positive and so~\cite{HSS} still applies.    These are classical tunneling ideas, see \cite{BCD1, BCD2}. 

At the time \cite{DSS} was being written, a representation of the spectral measure on the level of accuracy as obtained in \cite{CDST} for zero energy, had not yet been obtained. 
And therefore,  Mourre theory near the top of the barrier was used as a way to circumvent this difficulty. 
In this paper we close this gap and give a precise expansion of the resolvent near the potential barrier in the spirit of~\cite{CDST}. 
This allows for a more economical end result in~\cite{DSS}, but we do not write this out here since it only changes the number of angular derivates in the main dispersive estimate. 

In this paper we employ the   classical method of a {\em stretching of the independent variable}, known as Liouville-Green transform. 
This allows us as in \cite{Bleher} to reduce our equation to  the Weber equation which then becomes the leading normal form. 

Our main result concerns the standard {\em scattering matrix} $$\calS(E,\hbar) = \left(\begin{matrix} \calS_{11} & \calS_{12} \\ \calS_{21} & \calS_{22} \end{matrix}\right)$$ 
This matrix relates the incoming and outgoing Jost solutions at $+\infty$ and $-\infty$, respectively. See for example~\cite{Ramond} for the exact 
definition.  Due to the relations 
\[
\calS_{11}=\calS_{22},\qquad \calS_{12} = -\bar\calS_{21} \frac{\calS_{11}}{\bar\calS_{11}}
\]
it suffices in the following theorem to state results for $\calS_{11}, \calS_{21}$.

\bt\label{T1} Consider the Schr\"odinger equations
\eqref{eqf} with the potential $V\in L^1(\RR)$ smooth: $V\in C^v(\RR)$ with $v=\infty$ or $\omega$. 

Assume that $V(\xi)$ has a unique absolute maximum at $\xi=0$ with $V(0)=1, V'(0)=0,\ V''(0)=-2$.

There exist $\delta>0$ 
and $\beta=\beta(E){=O(1-E)}$ of class $C^v$ for $|1-E|<\delta$ so that the following quantities are the dominant behavior of the scattering coefficients in a sense made precise below.

Denote
\bel{valsAW}
A=e^{\pi \beta/(2\hb)},\ \ 
\ \theta=\frac{\beta}{2\hb}\left[ 1+\ln(2\hb/|\beta|)\right]+\arg\Gamma\left(\frac 12+\frac{i\beta}{2\hb}\right)
\ee

For $1-\delta<E\Le 1$ define
\bel{S11W}
{\mathcal{S}_{W,11}}=e^{\frac i{\hb}(I_+(E)+I_-(E))}\, e^{i \theta}\ \frac{1}{\sqrt{1+A^2}}
\ee

\bel{S21W}
{\mathcal{S}_{W,21}}={e^{\frac i{\hb}2I_-(E)}\, e^{i \theta}\ \frac{-iA}{\sqrt{1+A^2}}}
\ee

where $a<0<b$ are the two solutions of $E-V(\xi)=0$, and
\bel{defIp}
I_+(E):=\int_b^{+\infty}\left(\sqrt{E-V(\xi)}-\sqrt{E}\right)d\xi-b\sqrt{E}
\ee
\bel{defIm}
I_-(E):=\int_{-\infty}^a\left(\sqrt{E-V(\xi)}-\sqrt{E}\right)d\xi+a\sqrt{E}
\ee
\bel{defSdeE}
S(E)=\int_a^b\sqrt{V(\xi)-E}\,\, d\xi 
\ee

For $1<E<1+\delta$ define

\bel{S11W*}
\mathcal{S}_{W,11}=e^{ \frac i{\hb}\int_{-\infty}^{+\infty}\left(\sqrt{E-V(\xi)}-\sqrt{E}\right)d\xi}\, e^{i\theta}\ 
\frac{1}{\sqrt{1+A^2}}
\ee
\bel{S21nW}
\mathcal{S}_{W,21}=e^{ \frac{2i}{\hb}\int_{-\infty}^{0}\left(\sqrt{E-V(\xi)}-\sqrt{E}\right)d\xi}\  e^{\frac{i}{\hb}2\gamma^{-1}\phi_\omega}\  e^{i\theta}\ \frac{-iA}{\sqrt{1+A^2}}\ee
with $\gamma$ depending $C^v$ of $\a=1-E$, $\gamma=1+O(\a)$ and $\phi_{\omega}$ has an explicit expression in terms of the Taylor coefficients of the potential $V$ at $\xi=0$.

\

{\bf 1.} If $|1-E|/\hb\lesssim1$ then 
\bel{S11Tn0}
\mathcal{S}_{11}=\mathcal{S}_{W,11} \,(1+\hb\ln\hb\ e_{11}),\ \ \ \ \ \mathcal{S}_{21}=\mathcal{S}_{W,21}\,(1+ \hb\ln\hb \ e_{21})
\ee
where the error terms have the symbol-like behavior: if $\a=1-E$
\bel{symberrT}
|\partial_\a^ke_{ij}|<C_k\hb^{-k} \ \text{for all }k\in\NN
\ee

\

{\bf 2.} If $1-\delta<E<1$ and $h_1:={\hb}/(1-E)\ll 1$ then  
\bel{S11T}
\mathcal{S}_{11}=\mathcal{S}_{W,11} \,(1+h_1e_{11})=e^{\frac i{\hb}(I_+(E)+I_-(E))}\, e^{-S(E)/\hb} \,(1+h_1e'_{11})
\ee
\bel{S21T}
\mathcal{S}_{21}=\mathcal{S}_{W,21}\,(1+h_1e_{21})= -i e^{\frac i{\hb}2I_-(E)} \,(1+h_1e'_{21})
\ee
where the error terms $e_{11},\, e_{21}$ have the symbol-like behavior \eqref{symberrT}. 

\

{\bf 3.} If $1<E<1+\delta$ and $h_3:=\hb/(E-1)\ll 1$ then 
\bel{S11Tn}
\mathcal{S}_{11}=\mathcal{S}_{W,11} \,(1+h_3e_{11})=e^{ \frac i{\hb}\int_{-\infty}^{+\infty}\left(\sqrt{E-V(\xi)}-\sqrt{E}\right)d\xi}\,\left(1+h_3e'_{11}\right)
\ee
and
\bel{S21Tn}
\mathcal{S}_{21}=\mathcal{S}_{W,21}\,(1+h_3e_{21})=-ie^{ \frac{2i}{\hb}\int_{-\infty}^{0}\left(\sqrt{E-V(\xi)}-\sqrt{E}\right)d\xi}\  e^{-\frac{i}{\hb}2\gamma^{-1}\phi_\omega}\  \,\left(1+h_3e'_{21}\right)
\ee

The error terms $e_{11},\, e_{21}$ have the symbol-like behavior \eqref{symberrT}.

\et
\begin{figure}[!htb]
\centering
\includegraphics[scale=.3]{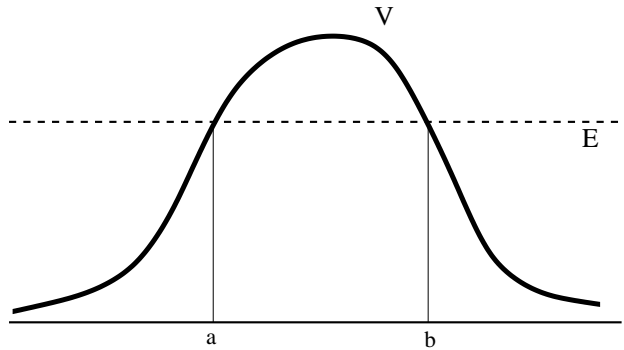}
\caption{$E<1$ with two turning points}
\end{figure}

As already mentioned, 
this theorem has many similarities with Theorem~1 of \cite{Ramond}, and the leading asymptotic behavior when $\hb/|E-1|\ll1$ reduce to the ones there. But there are also some crucial differences relating
to the way in which the error is represented. Ramond's theorem is an asymptotic result which allows for {\em additive errors} of the
form $O( e^{-\frac{\epsilon}{\hbar}})$. Such a representation of the resolvent is not amenable to the analysis of the long-term dispersive
decay of the wave or Schr\"odinger evolutions, as already mentioned above. The emphasis in our work is to represent all needed
quantities such as Jost solutions, scattering and connection coefficients, and the scattering matrix in the form 
\begin{equation}\label{main and error}
\text{main term}\times (1+\text{error})
\end{equation}
where the error is much smaller than~$1$ in size. This in itself is also not sufficient as the underlying oscillatory integrals which
arise in the time-dependent problem, see~\cite{DSS}, require smoothness of the error in the energy parameter  with {\em symbol-type} 
behavior of the derivatives (see the theorem and the subsequent sections for the precise meaning of this). 

Another distinct feature here is that these representations and error bounds hold in a neighborhood $|E-V(0)|+ \hbar \ll1$, and do
not in and of themselves constitute asymptotic results as they hold uniformly in that region. 

Finally, and in contrast to much (but not all) of the previous work, we do not restrict ourselves to analytic potentials but also allow
for the $C^{\infty}$ class. This leads to considerable technical effort with regard to the Liouville-Green transform which reduces
our problem to the Weber equation as a normal form.  This reduction is carried out in the following two sections. 
After that, we derive expressions for the incoming and outgoing Jost solutions of the form~\eqref{main and error}. 
This hinges on a careful perturbative analysis around canonical leading terms which are precisely given by the Airy and Weber equations. 
For the fine properties of the solutions to this perturbative analysis we analyze Volterra iterations as in~\cite{CDST}. 
The appendix of \cite{CDST} contains lemmas which precisely state the type of properties which we need, especially 
with regard to taking derivatives in the energy.  

Once a fundamental system is obtained, such as the outgoing and incoming Jost solutions,
we solve the connection problem and derive the scattering matrix. Of course, explicit representations of the resolvent are also
immediate at that point. 

There are three appendices. Appendix A summarizes properties of ultra-spherical polynomials since they come up in the $C^{\infty}$ analysis. 
Appendix B discusses the Weber equation, its standard fundamental system (parabolic cylinder functions), and the monodromy. 
Finally, Appendix C recalls the main perturbative results from~\cite{CDST} which play an important role here as well.

\section{Liouville transformation}

The modified parabolic cylinder functions will be called here Weber functions for short, and we call the Weber equation the differential equations that they satisfy (see \S\ref{WeberFunc}).

The following proposition provides the key normal form reduction to the Weber equation. 
It is very closely related to that of \cite{Bleher}, but with the main difference that we also need to establish regularity of the change of
variables in the phase space variables $(x,E)$ rather than in $x$ alone. Technically, this is considerably harder in the $C^\infty$ class and 
takes up most of the work in this section. 

\bp\label{Propo1}
For $1-\delta_1<E<1+\delta_1$ there exists an increasing $C^v$ function $\tilde{E}=\tilde{E}(E)$ with $\tilde{E}(1)=1$, and a function $\xi=\xi(y,E)$ defined for $|y|<\delta_2$, one to one, and of class $C^v$  in $(y,E)$ so that
\bel{eqxi}
\left[V(\xi)-E\right]\left(\frac{d\xi}{dy}\right)^2=1-y^2-\tilde{E}
\ee
Furthermore, $\xi(y,E)$ can be extend{ed} to an increasing function of $y$ of class $C^v$ on $\RR$ and
\begin{itemize}
\item for large $|y|$ we have 
\bel{asyxi}
\xi(y)= \pm \frac{1}{2\sqrt{E}}y^2\mp {\frac{1-\tilde E}{2\sqrt{E}}}\ln|y|\ {\pm}{C_\pm}+{o(1)}\qquad  y\to\pm\infty 
\ee
\item if $|V(\xi)|= V_1/\xi^{r+1}(1+o(1))$ for some $r>0$ as $\xi\to \infty$, then the error $o(1)$ in \eqref{asyxi} is $C_1/y^{2r}(1+o(1))$ with $C_1=V_1(2\sqrt{E})^r/(2rE)$.
\item 
if $V(\xi)$ behaves like a symbol, then $\xi(y,E)$ behaves like a symbol in the $y$ variable uniformly in $E$.
\item 
in the analytic case $\tilde{E}$ is unique with the properties described; in the $C^\infty$ case  $\tilde{E}$ is unique only for $E\Le 1$.
\end{itemize}
\ep

Throughout, we say that a  function $f\in C^\infty(\R)$ {\em behaves like a symbol} if each derivative gains one power of $y$ in terms of decay. 
We remark that the previous proposition also carries over the the case of finite regularity, but for the sake of simplicity we work in the infinitely differentiable regime.   

Strictly  speaking, equation \eqref{eqxi} is relevant only to small $y$. Technically speaking, however, it appears advantageous to
view it globally since this avoids partitioning the line in order to localize~\eqref{eqxi}.  

The proof is carried out in the next section.  We shall often write $\xi(y)$ instead of $\xi(y,E)$. 
Denoting $\beta=1-\tilde{E}$ and applying the  Liouville-Green transformation 
$$\psi(\xi(y))=\sqrt{\xi'(y)}{\psi_2}(y),$$ as well as \eqref{eqxi}, equation \eqref{eqf} becomes 
\bel{eqpsi2}
\frac{d^2{\psi_2}}{dy^2}(y)=\hb^{-2}(\beta-y^2){\psi_2(y)}+f(y){\psi_2(y)}
\ee
where {$f(y)=f(y;\beta)$} is the Schwarzian derivative
\begin{equation}\label{eq:S}
f=-\frac12 S[\xi]=\frac 34\left(\frac{\xi''}{\xi'}\right)^2-\frac 12\frac{\xi'''}{\xi'}
\end{equation}

{\bf Remark.} 
 It is easy to see that
\bel{estimf}
f(y)\sim \frac34 y^{-2}\ \ \ (y\to\pm\infty)
\ee
{and, if $\xi(y)$ behaves like a symbol, then so does $f(y)$.}

\section{Proof of Proposition\,\ref{Propo1}}\label{xiofy}

Sections \S\ref{smaly}-\S\ref{Cinf} establish the existence of the solution for small $y$. Its continuation to the whole real line is showed in \S\ref{xicont}. Its asymptotic behavior is established in \S\ref{par:asyxi}, and \S\ref{symbxi} shows symbol behavior.

\subsection{Existence of \texorpdfstring{$\xi(y)\ \text{for small }y$}{Lg} }\label{smaly}
 
If $\xi$ is small enough clearly there exists an increasing function $\xi(x)$ with the same regularity as $V$, $\xi(0)=0,\, \xi'(0)=1$ so that $V(\xi(x))=1-x^2$.
{Equation \eqref{eqxi}} becomes
\bel{eqyx}
(\beta-y^2)\left(\frac{dy}{dx}\right)^2=(\alpha-x^2)\omega(x)^2
\ee
where 
\bel{defomega}
\alpha=1-E,\qquad  \omega(x)=\frac{d\xi}{dx}
\ee

In the analytic case we show that:

\bp\label{prop:2}
Let $\omega$ be a function analytic at $0$, with $\omega(0)=1$. 
\begin{itemize}
\item There exist $\delta,\alpha_0>0$ and a unique function $\beta=\beta(\alpha)=\alpha+O(\alpha^2)$ analytic at $\a=0$, for which  equation (\ref{eqyx}) has a solution $y=y(x;\alpha)$ which is holomorphic in the polydisk $|x|<\delta,\, |\a|<\a_0$. Moreover, $x\mapsto y(x;\al)$ is a conformal map. 


\item Further requiring that this solution be close to the identity makes it unique and $y$ has the form $y=x+(\alpha-x^2)w(x;\alpha)$ with $w$ also holomorphic in the polydisk.
\end{itemize}
\ep

Similar results hold in the $C^\infty$ case:

\bp\label{Lemma1}
Let $\omega$ be a function of class $C^\infty$ near $0$, with $\omega(0)=1$.
\begin{itemize}
\item There exist $\delta,\alpha_0>0$ and a unique function $\beta=\beta(\alpha)=\alpha+O(\alpha^2)$, of class $C^\infty([0,\alpha_0])$, for which  equation (\ref{eqyx}) has a solution $y=y(x;\alpha)$ which is  $C^\infty([-\delta,\delta]\times [0,\alpha_0])$. Moreover, $x\mapsto y(x;\al)$ is a diffeomorphism. 
\item 
Further requiring that this solution be close to the identity makes it unique and $y$ has the form $y=x+(\alpha-x^2)w$ with $w$ having the same regularity as $y$.
\item 
Continuing $\beta(\alpha)$ within the class $C^\infty([-\alpha_0,\alpha_0])$ equation (\ref{eqyx}) has a solution $y=y(x;\alpha)$ which is $C^\infty$ on the rectangle $[-\delta,\delta]\times [-\alpha_0,\alpha_0]$ and remains a diffeomorphism in~$x$. 
\end{itemize}
\ep

The proofs of Propositions\,\ref{prop:2} and\,\ref{Lemma1} are developed in \S\ref{findbeta}-\S\ref{Cinf}. 
In addition, some technical material is relegated to Appendix~A.

{\subsection{The existence of $\beta=\beta(\alpha)$}\label{findbeta} 
We now derive a formula for $\beta(\a)$  in both the analytic and the $C^\infty$ cases.

  Let $ \alpha/\beta=\gamma, y=\sqrt{\beta/\alpha}\,\nv$. Equation \eqref{eqyx} becomes
 \begin{equation} \label{eq:eqov}
  (\alpha-\nv^2(x)){\nv '}^2(x)=\gamma^2\omega(x)^2(\alpha-x^2)
\end{equation}

It is easy to see that there are solutions of \eqref{eq:eqov} which are of class $C^v$ at $x=\sa$ and that they must satisfy $Y(\sa)=\pm\sa$. Similarly, there are solutions of \eqref{eq:eqov} which are of class $C^v$ at $x=-\sa$ and they must satisfy $Y(-\sa)=\pm\sa$. For generic $\gamma$ there are no solutions which are of class $C^v$ at both $x=\pm\sa$ (if $\a\ne 0$), but we will show that there exists a unique $\gamma$ (therefore, $\beta$) for which such a solution exists.

Consider $\a,x$ real with $\a>0$ and $|x|<\sa$. In this case  \eqref{eq:eqov} holds if
 \begin{equation} \label{eq:eqoY}
  \sqrt{\alpha-\nv^2(x)}\, {\nv '}(x)=\gamma\omega(x)\sqrt{\alpha-x^2}
\end{equation}
If $Y(x)$ satisfies \eqref{eq:eqov}, then so does $-Y(x)$; the formulation \eqref{eq:eqoY} chooses an increasing solution.

The solution $Y=Y(x)$ of \eqref{eq:eqoY} such that $Y(-\sa)=-\sa$ satisfies
\begin{equation}\label{eqY}
\int_{-\sa}^Y\sqrt{\a-s^2}\, ds=\int_{-\sa}^x\gamma \omega( s)\sqrt{\a-s^2}\, ds
\end{equation}
For this solution to be of class $C^v$ also at $x=\sa$ we must have $Y(\sa)=\sa$ and therefore, for $Y$ to be analytic at both $\sa$ and $-\sa$ we must have
\begin{equation}\label{eqfd}
\int_{-\sa}^{\sa}\sqrt{\a-s^2}\, ds=\int_{-\sa}^{\sa}\gamma \omega(s)\sqrt{\a-s^2}\, ds
\end{equation}
 This determines $\beta(=\alpha\gamma^{-1})$ as
 \bel{defbeta}
 \beta=\a\, \frac{\int_{-\sa}^{\sa} \omega( s) \sqrt{\a-s^2}\, ds}{ \int_{-\sa}^{\sa} \sqrt{\a-s^2}\, ds}=\frac {2}\pi \int_{-\sa}^{\sa} \omega( s) \sqrt{\a-s^2}\, ds
 \ee
 
Note that
\begin{equation}\label{simpbeta}
  \beta(\alpha)=\frac {2}\pi\int_{-\sa}^{\sa} \omega_{\rm even}(s)\sqrt{\alpha-s^2}\,ds=
\frac {4}\pi \int_{0}^{\sa}\omega_{\rm even}(s)\sqrt{\a-s^2} \, ds
\end{equation}
where $\omega_{\rm even}(x)=\frac 12 \omega(x)+\frac 12 \omega(-x)$.

\bl\label{regbeta}
If $\omega$ is analytic in a disk $|x|<\delta_1$ then $\beta(\a)$ defined by \eqref{defbeta} extends analytically to a disk $|\a|<\a_1$.

If $\omega$ is $C^\infty$ on an interval $|x|<\delta_1$ then $\beta(\a)$ defined by \eqref{defbeta} extends $C^\infty$ to an interval $\a\in[0,\a_1)$.
\el

\begin{proof}
Using Taylor polynomials of $\omega_{\rm even}$ at $0$ in \eqref{simpbeta} it is straightforward to show that $\beta$ extends of class $C^v$ at $\a=0$, and it satisfies $\beta(\alpha)=\alpha+O(\alpha^2)$ (since $\omega(0)=1$), and
thus $\gamma(\alpha)=1+O(\alpha)$ is also of class $C^v$ on $[0,\alpha_1)$ for some $\a_1>0$.
\end{proof}

The Taylor coefficients of ${\beta}/{\alpha}$ are obtained from the Taylor coefficients of $\omega(x)$. Indeed, 
if $$\omega(x)=1+\sum_{k=1}^{2n}\omega_kx^k+R_{2n}(x)$$ and $${\beta}/{\alpha}:=\gamma^{-1}=1+\sum_{k=1}^{n}\gamma_{2k}\a^{k}+S_{2n}(x),$$ then 
\bel{gamak}
\gamma_{2k}=\omega_{2k}\, \frac{2}{\pi}\,\int_{-1}^1t^{2k}\sqrt{1-t^2}\, dt\ \ \ \text{for all }k\ge0
\ee

{\bf Remark.} 
The function $\beta$ is nothing other than the classical action $S(E)$  (up to a normalizing factor) between the turning points for $1-E>0$ and small. Indeed, we have 
\bel{SdeE}
\frac 12 \pi\beta=\int_{-\sqrt{\beta}}^{\sqrt{\beta}}\, \sqrt{\beta-y^2}\, dy=\int_a^b\sqrt{V(\xi)-E}\,\, d\xi =S(E)
\ee
where $b=\xi(\sqrt{\beta}), \ a=\xi(-\sqrt{\beta})$.

\medskip

For $\alpha\Le 0$ we define $\beta(\alpha)$ to be a $C^v$ continuation of the function previously defined for $\a>0$ (unique in the analytic case).

\subsection{Recasting the differential equation in integral form}
In both the analytic and the $C^\infty$ cases we rewrite the differential equation as follows.

Consider first $\a,x$ real with $\a>0$ and $|x|<\sa$
and the form \eqref{eq:eqoY} of equation \eqref{eq:eqov}.

 Denote
\begin{equation}
  \label{defphi}
  \phi(x)=\int_{-\sa}^x\sqrt{\a-s^2}\,ds;\ \  h(x,\alpha) =\int_{-\sa}^x[ \omega(s)-\gamma^{-1}]\sqrt{\a-s^2}ds
\end{equation}

 A formal series expansion suggests a possible solution of \eqref{eq:eqoY}  of the form 
    $Y=x+O(\alpha,x^2)$. It is then natural to substitute $Y=x+v(x)$ in \eqref{eqY}. We then have
    \begin{equation}
      \label{eq:eqphi2}
      \phi(x+v)-\phi(x)=\gamma h(x;\a)
    \end{equation}
    
 Using the identity (Taylor polynomial with integral remainder):
 $$\phi(x+v)-\phi(x)=v\phi'(x)+\int_x^{x+v}(x+v-t)\phi''(t)\, dt$$
 equation \eqref{eq:eqphi2} becomes
 
 \begin{equation}
  \label{eq:eq2}
  v=\frac{h(x;\alpha)}{\sqrt{\a-x^2}}+\frac{1}{\sqrt{\a-x^2}}\int_x^{x+v}\frac{t(x+v-t)}{\sqrt{\a-t^2}}\, dt
\end{equation}
Further substituting $v(x;\a)=(\a-x^2)w(x;\a)$ and changing the variable of integration to $t=x+(\a-x^2)w\sigma$, equation 
\eqref{eq:eq2} becomes
\begin{equation}
  \label{eq:eq4}
  w=\gamma \, u(x;\alpha)+w^2\int_0^1\frac{(1-\sigma)[x+(\a-x^2)w\sigma]}{\sqrt{1-2xw\sigma +(x^2-\a)w^2\sigma^2 }}\, d\sigma :=\mathcal{N}(w)
\end{equation}
where $ u(x;\alpha)=(\a-x^2)^{-3/2}{h(x;\alpha)}$, that is, 
\begin{equation}\label{defu}
 u(x;\alpha)={(\a-x^2)^{-3/2}}\, \int_{-\sa}^x[ \omega(s)-\gamma^{-1}]\sqrt{\a-s^2}\,ds
 \end{equation}

Note that, in view of \eqref{eqfd} we also have
\bel{otheru}
 u(x;\alpha)={(\a-x^2)^{-3/2}}\, \int_{\sa}^x[ \omega(s)-\gamma^{-1}]\sqrt{\a-s^2}\,ds
 \end{equation}

 A crucial ingredient in the proof is that $ u(x;\alpha)$ extends of class $C^v$ in a neighborhood of $(0,0)$. The proof of Propositions\,\ref{prop:2} and\,\ref{Lemma1} will be completed by showing that the operator $\mathcal{N}$ in \eqref{eq:eq4} is contractive in appropriate functional spaces. Starting at this point, the analytic case and the $C^\infty$ case will be treated separately, in \S\ref{Analy} and \S\ref{Cinf},  respectively.

\subsection{Completing the proof of Proposition\,\ref{prop:2}}\label{Analy}

\begin{Lemma}\label{R1} Assume $\omega$ is analytic at $0$. Then the function $u(x;\alpha)$ defined by \eqref{defu} extends biholomorphically in a polydisk $|x|<\delta_1$, $|\alpha|<\a_1$.
 \end{Lemma}

\begin{proof}
Let $\delta_1,\a_1$ be small enough so that $\omega(x)$ and $\gamma^{-1}(\a)$ be analytic in the polydisk.

Note that the function $u$, initially defined for $\a>0$ and $-\sa<x<\sa$, satisfies the linear non-homogeneous equation
\begin{equation}
  \label{lineq}
  (\a-x^2)u'-3xu=\omega(x)-\gamma^{-1}
  \end{equation}
with coefficients depending analytically on the variable $x$ and the parameter $\a$. Therefore its solutions are analytic at all the regular points of the equation, namely at all $(x,\a)$ with $\a-x^2\ne 0$. Therefore \eqref{defu} extends homomorphically in the polydisk, outside the variety $\a-x^2=0$.

Note that functions of the type 
\begin{equation}\label{funtype}
F(\zeta)=\zeta^{-3/2}\int_0^\zeta f(t)\,t^{1/2}d\zeta 
\end{equation}
are analytic at $\zeta=0$ if $f(\zeta)$ is, as it is easily seen by a Taylor expansion of $f$ at $\zeta=0$. Therefore $u(x,\a)$ is analytic in $x$ at $x=-\sa$ (for $\a\ne 0$). Using \eqref{otheru} it follows that $u(x,\a)$ is analytic in $x$ also at $x=\sa$.

To show analyticity in $a=\sa$ at points with $\a=x^2\ne0$ we first clarify the analytic continuation of the formula \eqref{defu} to the complex domain. For $\sa>0$, $a=\sa$, $|x|<a$ changing the variable of integration to $s=at$  we obtain
\begin{multline}
u(x,\alpha):=\tilde{u}(x;a)\\
=\left(1-\frac xa\right)^{-3/2}\left(1+\frac xa\right)^{-3/2}\int_{-1}^{x/a}\frac{\omega(ta)-\gamma^{-1}(a^2)}{a}\sqrt{1-t}\,\sqrt{1+t}\, dt \\
:=\left(1-z\right)^{-3/2}\left(1+z\right)^{-3/2}\int_{-1}^{z}\frac{\omega(ta)-\gamma^{-1}(a^2)}{a}\sqrt{1-t}\sqrt{1+t}\, dt \\
:=\tilde{\tilde{u}}(z,a)\ \ \text{where }z=\frac xa
\end{multline}
which we can continue to complex $(z,a)$. This function is manifestly analytic in $a$ (recall that $\omega(0)=1=\gamma(0)$) and it is analytic in $z$, including at $z=-1$ (it has the form \eqref{funtype} for $\zeta=z+1$) and, using \eqref{otheru}, also at $z=1$. Since we showed analyticity in $x$, it follows that $\tilde{u}(x;a)$ is analytic in $a$ for $a\ne 0$.

The function  $\tilde{u}(x;a)$ is even in $a$ (by \eqref{eqfd}), and therefore $u(x,\alpha)$ is analytic in $\alpha=a^2$ for $\a\ne 0$. 

Points $(x,0)$ with $x\ne 0$ are also regular points of the equation \eqref{lineq}, therefore solutions are analytic at these points.

It only remains to show analyticity of $u(x,\a)$ at $(0,0)$. This holds
by Hartog's extension theorem, since $u(x;\alpha)$ is analytic in a punctured polydisk.
\end{proof}

Proposition~\ref{prop:2} now follows from the following result.

\begin{Lemma}\label{L1}
Consider the Banach space $\mathcal{B}$ of functions analytic in the polydisk  $|x|<\delta,|\a|<\a_0$, continuous up to the boundary, with the sup norm.

There exist $\delta,\a_0$ small enough and $R>0$ so that the operator $\mathcal{N}$ defined by \eqref{eq:eq4} leaves invariant the ball of radius $R$, $\mathcal{B}_R\subset \mathcal{B}$ and it is a contraction there.

As a consequence \eqref{eq:eq4} has a unique solution in  $\mathcal{B}_R$.
 \end{Lemma}
 
\begin{proof}
  Let $$M=\sup\{|\gamma u(x;\a)|\, |\, |x|\leq \delta_1,|\a|\leq\a_1\}$$ (after possibly lowering $\delta_1, \a_1$ so that $u$ is continuous up to the boundary of the polydisk). We will look for $\delta\leq \delta_1,\ \a_0\leq \a_1$, $R>0$ with the properties stated in the lemma.
Let $w\in\mathcal{B}_R$, $|x|\leq\delta$, $|\a|\leq\a_0$.
 
To ensure that $\mathcal{N}$ is well defined we require that $R,\delta,\a_0$ satisfy
 \begin{equation}\label{con1}
 2\delta R+(\delta^2+\a_0)R^2<1-c_0^2\ \ \text{for some }c_0\in(0,1)
 \end{equation}
which implies that $|1-2xw\sigma +(x^2-\a)w^2\sigma^2 |>c_0^2>0$.
 
 Using the estimate $|x+(\a-x^2)w\sigma|\leq \delta +(\delta^2+\a_0)R$ we see that $\mathcal{N}$ leaves the ball $\mathcal{B}_R$ invariant if
 \begin{equation}\label{con2}
 R^2\frac{\delta +(\delta^2+\a_0)R}{c_0}+M\leq R
 \end{equation}
 The contractivity of $\mathcal{N}$ follows if we show that $|\partial\mathcal{N}/\partial w|<1$ for all $|x|\leq\delta,|\a|\leq\a_0, |w|\leq R$, which holds if
 \begin{equation}\label{con3}
 2R\frac{\delta +(\delta^2+\a_0)R}{c_0}+R^2\frac{(\delta^2+\a_0)}{c_0}+R^2\frac{[\delta +(\delta^2+\a_0)R]^2}{c_0^3}<1
 \end{equation}
 Clearly the conditions \eqref{con1}, \eqref{con2}, \eqref{con3} hold  if $\delta,\a_0$ are small enough. For example, let $c_0=1/2$, $R=4M/3$, $\delta=1/(32 R)$, $\a_0=1/(16R^2)$. 
 \end{proof}
 
 \subsubsection{The analytic continuation of $u(x,\a)$ to $\a<0$}\label{ancase}
 This section motivates the choice, in the $C^\infty$ case, of the definition of $u(x,\a)$ for $\a<0$ in \eqref{defumR}.

{Consider $x$ in a small disk centered at $0$ where $\omega$ is analytic. } We re-write \eqref{eq:eqoY} as 
\bel{inteq}
{\int_x^{\nv}{\sqrt{s^2-\a}}\,ds=\int_{\sqrt{\alpha}}^x\, [\gamma\omega(s)-1]\sqrt{s^2-\a}} \, ds \ \ \ \ 
\ee

{Upon analytic continuation of  \eqref{inteq} in $\a$, going counterclockwise along half a circle of radius $|\a|$ to $\a<0$ we have $\sa=i\sqrt{-\a}$ where $\sqrt{-\a}>0$ and we write the right-hand side of \eqref{inteq} as
$$\int_{i\sqrt{-\a}}^x\, [\gamma\omega(s)-1]\sqrt{s^2-\a} \, ds=\phi_\om(\a)+\int_0^x\, [\gamma\omega(s)-1]\sqrt{s^2-\a} \, ds$$
where
\begin{multline}\label{defcphi}
\phi_\om(\a)=\int_{i\sqrt{-\a}}^0\, [\gamma\omega(s)-1]\sqrt{s^2-\a} \, ds=i\gamma\int_{\sqrt{-\a}}^0\, [\omega(i\tau)-\gamma^{-1}]\sqrt{(-\a)-\tau^2} \, d\tau\\
=i\a \gamma \int_0^1[\omega(i t\sqrt{-\a})-\gamma^{-1}] \, \sqrt{1-t^2} \, dt 
\end{multline}
Expanding $\omega$ and $\gamma^{-1}$ in Taylor series and using the explicit form of the Taylor coefficients of $\gamma^{-1}$ in \eqref{gamak} the last integral simplifies to
\bel{defcphi2}
\phi_\om(\a)=i\a \gamma \int_0^1\omega_{\text{odd}}(i t\sqrt{-\a}) \, \sqrt{1-t^2} \, dt \\
:=\a \sqrt{-\a}\phi_1(\a)
\ee
where $\phi_1(\a)$ is real-valued for $\a$ real, and analytic.

Note that clockwise continuation of  \eqref{inteq} to $\a<0$ gives the same value: the integral becomes
$$\int_{-i\sqrt{-\a}}^0\, [\gamma\omega(s)-1]\sqrt{s^2-\a} \, ds=-i\a \gamma \int_0^1\omega_{\text{odd}}(-i t\sqrt{-\a}) \, \sqrt{1-t^2} \, dt=\phi_\om(\a)$$

\subsection{Completing the proof of Proposition\,\ref{Lemma1}}\label{Cinf}

We  need to extend $\beta(\a)$ and $u(x;\a)$ for $\a<0$.
We first define  $\beta(\a)$ as any $C^\infty$ continuation from $[0,\a_1)$ to $(-\a_1,\a_1)$.  Next, we extend $\phi_{\om}$. 

\begin{Lemma}  \label{oddofodd}
Assume that $\omega$ is $C^{\infty}$. Then the function $\phi_{\om}$ from \eqref{defcphi2} admits a $C^{\infty}$ extension to negative~$\alpha$. 
\end{Lemma}
\begin{proof}
We have 
\bel{notomodd}
\omega_{\text{odd}}(x)=x\tilde{\omega}_{\text{even}}(x)=xg_\omega(x^2)\ \ \text{where } g_\omega\in C^\infty([0,{\delta^{2}}])
\ee
Take $g_\omega^c$ any $C^\infty([{-\delta^{2},\delta^{2}}])$ continuation of $g_\omega$ and define $\omega_{\text{odd}}(ix)=ixg_\omega^c(-x^2)$ which is in $C^\infty([-\delta,\delta])$. Note that $i\omega_{\text{odd}}(ix)$ is real-valued, and its Taylor polynomial approximations at $x=0$ are obtained by replacing $x$ by $ix$ in the Taylor polynomial approximations at $x=0$ of $\omega_{\text{odd}}(x)$, multiplied by $i$.
\end{proof}

We extend $u(x,\a)$ by means of the formulas we obtained in the analytic case:
\bel{defumR}
u(x;\alpha)=\left\{  \begin{array}{lll} 
\ (\alpha-x^2)^{-3/2}\, \int_{\pm\sqrt{\alpha}}^x\, ds\, [\omega(s)-\gamma^{-1}]\sqrt{\alpha-s^2} &\text{for }& -\sqrt{\alpha}<x<\sqrt{\alpha}\\ \\
-(x^2-\alpha)^{-3/2}\, \int_{\sqrt{\alpha}}^x\, ds\, [\omega(s)-\gamma^{-1}]\sqrt{s^2-\alpha} &\text{for }&{ \delta>}x>\sqrt{\alpha}\\ \\
-(x^2-\alpha)^{-3/2}\, \int_{-\sqrt{\alpha}}^x\, ds\, [\omega(s)-\gamma^{-1}]\sqrt{s^2-\alpha} &\text{for }&{-\delta<} x<-\sqrt{\alpha}\\ \\
-(x^2-\alpha)^{-3/2}\,\left\{ \gamma^{-1} \phi_\om(\a)+\int_{0}^x\, ds\, [\omega(s)-\gamma^{-1}]\sqrt{s^2-\alpha} \right\}&\text{for }& \alpha\Le 0
\end{array}\right.
\ee

Note that we have
 \bel{valsu1}
u(0,0)=\frac{- \tilde{\omega}(0)}{3},\ \ u(\pm a,\a)=\frac{\mp 1}{3a}\left[\omega(\pm a)-\gamma^{-1}\right]\ \ \text{if }\a>0,\,a=\sa
\ee
where $\omega(x)=1+x\tilde{\omega}(x)$.

\begin{Lemma}\label{R2} The function $u(x;\alpha)$ defined by \eqref{defumR} is $C^\infty$ in the neighborhood of $(0,0)$ given by $|x|<\delta$, $|\alpha|<\a_0$.
 \end{Lemma}
 
 \begin{proof}
All the steps in the proof of Lemma\,\ref{R1} hold in the $C^\infty$ case as well, except for the regularity at $(0,0)$ which is proved in \S\ref{Pfan00}. For the moment let us assume Lemma\,\ref{R2} holds and  complete the proof of Proposition\,\ref{Lemma1}.
\end{proof}

\begin{Lemma}\label{L1C}
Consider the Banach space $\mathcal{B}$ of functions continuous in the rectangle $[-\delta,\delta]\times[-\a_0,\a_0]$ with the sup norm.

There exist $\delta,\a_0$ small enough and $R>0$ so that the operator $\mathcal{N}$ defined by \eqref{eq:eq4} leaves invariant the ball of radius $R$, $\mathcal{B}_R\subset \mathcal{B}$ and satisfies $|\partial_w\mathcal{N}(w)|<1$, therefore it is a contraction.

As a consequence \eqref{eq:eq4} has a unique solution in  $\mathcal{B}_R$.

 \end{Lemma}
 \begin{proof}
 The same arguments as in the proof of Lemma\,\ref{L1} also establish Lemma\,\ref{L1C}.  
 \end{proof}
 
 \begin{Lemma}\label{L2C}
 The continuous function $w=w(x,\a)$ satisfying $w=\mathcal{N}[w]$ given by Lemma\,\ref{L1C} is in fact $C^\infty$.
 \end{Lemma}

\begin{proof}
Consider the function $\Phi(x,\a,w)=w-\mathcal{N}(w)$. Lemma\,\ref{L1C} shows that for each $(x,\a)\in[-\delta,\delta]\times[-\a_0,\a_0]$ equation  $\Phi(x,\a,w)=0$ defines implicitly $w=w(x,\a)\in \mathcal{B}_R$. We have $\partial_w\Phi(x,\a,w)=1-\partial_w\mathcal{N}(w)\ne 0$ (since  $|\partial_w\mathcal{N}(w)|<1$ by Lemma\,\ref{L1C}).  
\end{proof}

\subsubsection{Regularity of $u(x,\alpha)$ at $(0,0)$ in the $C^\infty$ case}\label{Pfan00}

We  expand the integrand, and the function $u$, in Gegenbauer polynomials. See \S\ref{regatta}  for an overview of the properties of these polynomials that we use.

We write formula \eqref{defumR} in operator notation as $u(x;\alpha):=\mathcal{J_\a}(\omega-\gamma^{-1})$. Note that for $\alpha>0$ we have
\begin{multline}
\mathcal{J_\a}f(x,\a)=\frac{1}{(\a-x^2)\,|\a-x^2|^{1/2}}\, \int_{-\sa}^x\, ds\, f(s,\a)\sqrt{|\a-s^2|}\\
=\frac{1}{(1-(\frac x\sa)^2)\,|1-(\frac x\sa)^2|^{1/2}}\, \int_{-1}^{x/\sa}\, d\tau\, \frac{f(\sa\tau,\a)}{\sa}\sqrt{|1-\tau^2|}
\end{multline}
smoothly continuable for $x=\pm \sa\ne 0$ if $f=\omega-\gamma^{-1}$.

From \eqref{cintU} we see that, for $\a>0$,
\bel{cintUx}
{C}^{(2)}_{n-1}\left(\frac x\sa\right)=\frac{-n(n+2)}{2}\, \sa\,\mathcal{J_\a}\left[{U}_n\left(\frac x\sa\right)\right]
\ee
For $\a\Le 0$ the operator $\mathcal{J_\a}$ is given by the formula
\bel{defJaneg}
\mathcal{J_\a}f(x,\a)=\frac{-1}{(x^2-\a)^{3/2}}\,\left[ \gamma^{-1}\phi_f(\alpha)+ \int_{0}^x\, ds\, f(s,\a)\sqrt{s^2-\a}\right]
\ee
where $\phi_f$ is defined by \eqref{defcphi2} and Lemma~\ref{oddofodd}.  While the extension of $\phi_{f}$ to negative
arguments is
not unique in the $C^{\infty}$ case, its Taylor expansion at $0$ is unique. And for the regularity at $(0,0)$ that is all we need in this proof.

\subsubsection{Preparatory remarks}\label{preprem}

Henceforth, {\em smooth} means infinitely differentiable. 

\begin{enumerate} 
\item[(A)] Consider the differential equation 
\bel{oeq5}
(\a-x^2)y'(x)-3xy(x)=f(x)
\ee
where $f$ is smooth. 
\begin{enumerate}
\item[(i)] If $\a>0$ the equation has at most one solution which is smooth at both $x=\pm \sa$.
\item[(ii)] If $\a< 0$ there is a unique smooth solution with a specified initial condition at $x=0$: $y(0)=y_0$ for any given $y_0$.
\item[(iii)] For $\al=0$ if $f(0)=0$ then there is a unique solution smooth at $x=0$. 
\item[(iv)] $y=\mathcal{J_\a}f$ solves the equation \eqref{oeq5}.
\end{enumerate}

\medskip

\item[(B)] In equation \eqref{eqCUnt}, changing the independent variable to $t=\frac x\sa$ we obtain that ${C}^{(2)}_{n-1}(\frac x\sa)$ satisfies \eqref{oeq5} with $$f(x)=\frac{-n(n+2)}{2}\, \sa\,{U}_n(\frac x\sa)$$ Since ${C}^{(2)}_{n-1}$ are odd functions if $n$ is even, they vanish at the origin, and by analytic continuation in $\a$ we see that \eqref{cintUx} is valid also for $\a< 0$. If $n$ is odd in that formula, then the choice of branch of $\sqrt{\a}$ is immaterial. For $n$ even we need to choose the same branch on both sides.
\end{enumerate}

\medskip

We split the integrand in formula \eqref{defumR} into even and odd parts: $$\omega(x)-\gamma ^{-1}=[\omega_{even}(x)-\gamma ^{-1}]+\omega_{odd}(x)$$ and show that $u$ is a sum of two $C^\infty$ functions $u=u_e+u_o$ where $u_e=\mathcal{J_\a}(\omega_{even}-\gamma ^{-1})$ and $u_o=\mathcal{J_\a}\omega_{odd}$.
The strategy is to Tayor expand the functions $\omega_{even}-\gamma ^{-1}$, $\omega_{odd}$, respectively,
and then to apply $\mathcal{J_\a}$ to this expansion. The operator takes the polynomial part onto another polynomial in both $x$ and $\alpha$, whereas
the Taylor remainder is estimated by hand.  We will distinguish between $\al>0$, $\al=0$ and $\al<0$ throughout. It is worth noting that the calculations involving polynomials are insensitive to the choice of sign in $\al$, since they only involve analytic functions.

\medskip

{\bf The even part.} 
In this case we work with even functions $f$, therefore $\phi_f(\a)=0$ in \eqref{defJaneg}.
 Let
$$f_{2k+2}(x,\a)=(2k+4)\, (\a-x^2)^{k+1}-\a(2k+3)(\a-x^2)^k$$
Using (A) above  it is easy to check that
\bel{trikform}
\mathcal{J_\a}f_{2k+2}(x,\a)=x(\a-x^2)^k
\ee
since both functions solve \eqref{oeq5} with $f=f_{2k+2}$ and (i) for $\a>0$ they are smooth at $x=\pm \sa$, and (ii) if $\a\Le 0$, both are $0$ at $x=0$.

It follows that for $\a>0$
\bel{fLv}
\int_{-\sa}^{\sa} f_{2k+2}(s,\a)\sqrt{\a-s^2}\, ds=0
\ee
Re-write with $t=\frac x\sa$, 
$$f_{2k+2}(x,\a)=\a^{k+1}\left[ (2k+4)\, (1-x^2/\a)^{k+1}-(2k+3)(1-x^2/\a)^k\right]:=\a^{k+1}\phi_{2k+2}(t)$$
By \eqref{fLv} we have
$\int_{-1}^1\phi_{2k+2}(t)\sqrt{1-t^2}\,dt=0$
which means that the polynomial $\phi_{2k+2}$ belongs to the span of the Chebyshev polynomials of the second kind, i.e., to  $$\mathrm{Sp}(U_2,U_4,\ldots, U_{2k+2})$$ 
Indeed,  $\phi_{2k+2}$ has a zero component along $U_0\equiv 1$, as well as along all odd $U_{j}$.

\subsubsection{Taylor approximation}

Write the Taylor polynomial
\bel{tpo}
\omega_{even}(x)-1=\sum_{k=1}^n\omega_{2k}\,x^{2k}+R_{2n}(x),\ \text{where }R_{2n}(x)=\frac{\omega^{(2n+2)}(\xi)}{(2n+2)!}x^{2n+2},\ \xi\in[0,x]
\ee
 Therefore, for $|x|\Le\delta$ we have $|R_{2n}|\leq C_n|x|^{2n+2}$.
Similarly
\bel{tgo}
\gamma^{-1}(\al)-1=\sum_{k=1}^n\gamma_{2k}\a^{k}+S_{2n}(\al)
\ee
where $|S_{2n}|\leq C_n\a^{n+1}$ for $|\a|<\a_0$.

\medskip

{\em For $\a\ne 0$:}

We have, from \eqref{eqfd}, that 
$$\int_{-1}^1[\omega_{even}(\sa t)-\gamma ^{-1}]\, \sqrt{1-t^2}\,dt=0$$
and by \eqref{gamak} we have
$$\int_{-1}^1(\omega_{2k}t^{2k}-\gamma_{2k})\, \sqrt{1-t^2}\,dt=0,\ \ \  \text{for each }k\Ge 1$$
therefore $\omega_{2k}t^{2k}-\gamma_{2k}$ has zero component along $U_0(t)$: $$\omega_{2k}t^{2k}-\gamma_{2k}=\sum_{j=1}^k c_{2j,2k}U_{2j}(t)$$ implying that 
$$P_{2n}(x,\a)=\sum_{k=1}^n(\omega_{2k}x^{2k}-\gamma_{2k}\a^{k})=\sum_{k=1}^n \a^{k}\sum_{j=1}^k c_{2j,2k}U_{2j}(x/\sa)\ \ \text{is a polynomial in } (x^2,\a)$$ 

Using \eqref{tpo}, \eqref{tgo}, \eqref{eqCUnt} we find that
$$v_e=\mathcal{J_\a}(P_{2n})+\mathcal{J_\a}(R_{2n}-S_{2n})$$
where, by \eqref{cintUx} and (B),
\begin{align*}
\mathcal{J_\a}(P_{2n})(x) &= \sum_{k=1}^n \mathcal{J_\a} (\omega_{2k}x^{2k}-\gamma_{2k}\a^{k}) \\
& =\sum_{k=1}^n \sa^{2k-1}\sum_{j=1}^k c_{2j,2k} \frac{2}{-2j(2j+2)}\,{C}^{(2)}_{2j-1}(x/\sa)
\end{align*}
which is a polynomial in $(x,\a)$ of degree $2n-1$ in $x$ and $n-1$ in $\alpha$ (and real-valued, even for $\a<0$). 
In fact, it is a real-linear combination of the monomials $x^{2\ell-1}\al^{k-\ell}$, $1\le \ell \le k\le n$. 

We will next show that
\bel{remeven}
\mathcal{J_\a}(R_{2n}-S_{2n})=O(x^{2n+1})+O(x^{2n-1}\a)+\ldots+O(|\a|^{n+\frac 12})
\ee
which implies that $u_{\text{o}}$ is of class $C^{n-1}$ at $(0,0)$ for any $n$.

\subsubsection{Estimating the remainder}\label{estrem} 

Ideally, \eqref{remeven} should be obtainable using convergence theorems  of Gegenbauer series. Interestingly though, in our case we use approximations of functions by Gegenbauer polynomials on intervals that exceed the interval of orthogonality where they would classically be known to hold.

\medskip

I. {\em Estimate for $\a>0$.} Denote $a=\sa$.

\medskip

Fix $\lambda\in(\frac 12,1)$. The proof  splits into several cases. Throughout, constants  $C_{n}$ depend only on $n$ and can change from line to line. 

\smallskip 

{\em 1 a: for $|x|\Le \lambda a$} we note that $$\int_{-a}^a P_{2n}\sqrt{a^2-s^2}\,ds=0$$ implies $$0=\int_{-a}^a (R_{2n}-S_{2n})\sqrt{a^2-s^2}\, ds=2\int_{0}^a (R_{2n}-S_{2n})\sqrt{a^2-s^2}\,ds$$ and therefore $$\int_{a}^x (R_{2n}-S_{2n})\sqrt{a^2-s^2}\,ds=\int_{0}^x (R_{2n}-S_{2n})\sqrt{a^2-s^2}\,ds$$ 

Finally, 
\begin{align*}
\big| \mathcal{J_\a}(R_{2n}-S_{2n})(x)\big| &=(a^2-x^2)^{-3/2}\big| \int_{0}^x (R_{2n}-S_{2n})\sqrt{a^2-s^2}\, ds\big|\\
&\Le C_n(a^2-x^2)^{-3/2} \int_{0}^x (|s|^{2n+2}+a^{2n+2})\sqrt{a^2-s^2}\, ds\\
&\Le C_n(1-\lambda^2)^{-3/2} \int_{0}^{x/a} (a^{2n+1}|t|^{2n+2}+a^{2n+1})\sqrt{1-t^2} \,dt\\
&\Le C_n \left[a^{2n+1}(x/a)^{2n+3}+a^{2n+1}(x/a)\right] \Le C_n a^{2n+1}
\end{align*}

\medskip

{\em 1 b: for $ \lambda a<x<a$} we estimate  (changing the integration variables first to $s=at$ and then to $1-t=\tau(1-x/a)$)

\begin{align*}
\big| \mathcal{J_\a}(R_{2n}-S_{2n})(x)\big| &=(a^2-x^2)^{-3/2}\big| \int_{a}^x (R_{2n}-S_{2n})\sqrt{a^2-s^2}\,ds\big|\\
&\Le C_n(1-x/a)^{-3/2} \int_{x/a}^{1} t^{2n+2}a^{2n+1}\sqrt{1-t^2}\,dt\\
&= C_n \int_{0}^{1}  a^{2n+1}\sqrt{2+\tau(1-x/a)} \sqrt{\tau}\,d\tau
 \Le C_{n} a^{2n+1}
\end{align*}

\medskip

{\em 2 b: for $a<x\Le 2 a$} the estimate is similar to the case {\em 1b}. Indeed,
\begin{multline}\label{sttt}
\big| \mathcal{J_\a}(R_{2n}-S_{2n})(x)\big|=(x^2-a^2)^{-3/2}\big| \int_{a}^x (R_{2n}-S_{2n})\sqrt{s^2-a^2}\, ds\big|\\
\Le  C_n(x^2/a^2-1)^{-3/2} \int_{1}^{x/a} (t^{2n+2}a^{2n+1}+a^{2n+1})\sqrt{t^2-1}\, dt\\
= C_n \int_{0}^{1}  a^{2n+1} \sqrt{2+\tau(x/a-1)} \sqrt{\tau}\, d\tau \Le C_n a^{2n+1}
\end{multline}

\medskip

{\em 2 a: for $2a\Le x<\delta$} in the second line of the estimate \eqref{sttt} we use $\sqrt{t^2-1}<t$ and obtain
\begin{align*}
\big| \mathcal{J_\a}(R_{2n}-S_{2n})(x)\big|
&\Le  C_n(x^2/a^2-1)^{-3/2} \int_{1}^{x/a} a^{2n+1}(t^{2n+3}+t)\,dt\\
&\Le C_n  (x^2/a^2-1)^{-3/2}  a^{2n+1} \left[ x^{2n+4}/a^{2n+4}+x^2/a^2\right]\\
&\Le C_n(x^{2n+1}+a^{2n+1})
\end{align*}

\medskip

{\em 3. for $x<0$} the estimates are similar.

\bigskip

II. {\em Estimate for $\a< 0$}  

\medskip

We have
\begin{align*}
\big| \mathcal{J_\a}(R_{2n}-S_{2n})(x)\big|&=(x^2+a^2)^{-3/2}\big| \int_{0}^x (R_{2n}-S_{2n})\sqrt{s^2+a^2}\, ds\big|\\
&\Le  C_n\, \,(x^2+a^2)^{-3/2} \int_{0}^{|x|} (s^{2n+2}+a^{2n+2})\sqrt{s^2+a^2}\, ds\\
&=a^{2n+1}(x^2/a^2+1)^{-3/2}\int_0^{|x|/a}(t^{2n+2}+1)\sqrt{1+t^2}\, dt:=\mathcal{E}_n
\end{align*}
For $|x|/a\Le1$ we have 
$$\mathcal{E}_n\le Ca^{2n+1}\int_0^{|x|/a} 1\, dt\le Ca^{2n}|x|$$
while for  $|x|/a>1$ we estimate
\begin{align*}
\mathcal{E}_n &= a^{2n+1}(x^2/a^2+1)^{-3/2}(\int_0^1+\int_1^{|x|/a})(t^{2n+2}+1)\sqrt{1+t^2}\, dt\\
&\le C\, a^{2n+1}+ C\, a^{2n+1}\frac{a^3}{|x|^3}\int_1^{|x|/a}t^{2n+3}\, dt\\
&\le C\, a^{2n+1}+C\, |x|^{2n+1} 
\end{align*}

\bigskip

{\bf The odd part.} Denote
$$f_{2k+1}(x,\a)=-(2k+3)\, x(\a-x^2)^{k}$$
Using (A)  of \S\ref{preprem} it is easy to check the equality
$$\mathcal{J_\a}f_{2k+1}(x,\a)=(\a-x^2)^k$$
since both functions solve \eqref{oeq5} with $f=f_{2k+1}$ and (i) for $\a>0$ both are smooth at $x=\pm \sa$, and (ii) if $\a\Le 0$, both equal $\a^k$ at $x=0$.

\smallskip

To obtain the Taylor approximations we write the Taylor polynomial 
\begin{equation}
\begin{aligned}
\label{tayodd}
\omega_{odd}(x) &=\sum_{k=1}^n\omega_{2k-1}x^{2k-1}+R_{2n-1}(x),\ \text{where }\\
R_{2n-1}(x) &=\frac{\omega^{(2n+1)}(\xi)}{(2n+1)!}x^{2n+1},\ \xi\in[0,x]
\end{aligned}
\end{equation}
 therefore, for $|x|\Le\delta$ we have $|R_{2n-1}|<C_n|x|^{2n+1}$.

\medskip

{\em I. The case $\a>0$.}

\medskip

Rewriting
$$f_{2k+1}(x,\a)=\sa^{2k+1}\left[ -(2k+3)\frac x\sa(1-x^2/\a)^k\right]:=\sa^{2k+1}\phi_{2k+1}(t),\ \ t=\frac x\sa$$
 the polynomial $\phi_{2k+1}$ belongs to the span of the Chebyshev polynomials of the second kind $$\mathrm{Span}(U_1,U_3,\ldots, U_{2k+1}) $$
Expand the Taylor approximation of $\omega_{odd}(x)$ as
\bel{exanegP}
P_{2n-1}(x):=\sum_{k=1}^n\omega_{2k-1}x^{2k-1}=\sum_{k=1}^n \sa^{2k-1}\sum_{j=1}^k c_{2j-1,2k-1}U_{2j-1}(x/\sa)
\ee

We have
$u_o=\mathcal{J_\a}(P_{2n-1})+\mathcal{J_\a}(R_{2n-1})$
where, by \eqref{cintUx},
\bel{exanegJ}
\mathcal{J_\a}(P_{2n-1})=\sum_{k=1}^n\omega_{2k-1}x^{2k-1}=\sum_{k=1}^n \a^{k-{1}}\sum_{j=1}^k c_{2j-1,2k-1} \frac{2}{-(2j-1)(2j+1)}\,{C}^{(2)}_{2j-2}(x/\sa)
\ee
which is a polynomial in $({x^2,\a})$ of degree $2n-2$ in $x$.  
 
 \medskip
 
 {\em II. The case $\a<0$.}
 
 \medskip
 
 Relations \eqref{exanegP}, \eqref{exanegJ} still hold, by analytic continuation (note continuation clockwise or counterclockwise yield the same result).
 
 \smallskip
 
We will next show that 
\bel{oremeven}
\mathcal{J_\a}(R_{2n-1})=O(x^{2n})+O(x^{2n-2}\a)+\ldots+O(\a^n)
\ee
which implies that $u(x,\a)$ is of class $C^{n-1}$ at $(0,0)$ for all $n$.

\subsubsection{Estimate of the remainder} 

{\em I. The case $\a>0$:} 

\medskip

For $x>\sa$ (and with $a=\sa$) we have
\begin{align*}
\big| \mathcal{J_\a}R_{2n-1}(x)\big| &=(x^2-\a)^{-3/2}\big| \int_{\sa}^x {R}_{2n-1}(s)\sqrt{s^2-\a}\, ds\big|\\
& \Le C_n (x^2-\a)^{-3/2} \int_{\sa}^x  s^{2n+1} \sqrt{s^2-\a}\, ds  \\
&\Le
 C_n\, x^{2n}(x-a)^{-\frac32} \int_{a}^x   \sqrt{s-a}\, ds   
\Le C_{n}\, x^{2n}
\end{align*}
The case $x<-\sa$  is analogous, while for $|x|<\sa$ we obtain $\big| \mathcal{J_\a}R_{2n-1}(x)\big|\Le C_{n}\,\a^{n}$.

\medskip

{\em II. For $\al\Le 0$} denote $a=\sqrt{-\a}$. 

\medskip

Since both $\omega_{odd}$ and its Taylor approximation $P_{2n-1}$ are odd functions, so is the remainder $R_{2n-1}$. With the notation \eqref{notomodd} we have $$R_{2n-1}(x)=x^{2n+1}g_{R_{2n-1}}(x^2)$$ where $g_{R_{2n-1}}\in C^\infty([0,\delta^{1/2}))$. Consider $g_{R_{2n-1}}^c$ a continuation in $ C^\infty((-\delta^{1/2},\delta^{1/2}))$.  We have $|g_{R_{2n-1}}^c|\Le C$ on $[-\delta^{1/2},\delta^{1/2}]$. 

 By \eqref{defcphi2} we have
$$\gamma^{-1}\phi_{R_{2n-1}}(\a)=i\a  \int_0^1   (i t\sqrt{-\a})^{2n+1}g_{R_{2n-1}}(t^2\a) \, \sqrt{1-t^2} \, dt $$
and using \eqref{defJaneg} we have
\begin{multline}
\big| \mathcal{J_\a}g_{R_{2n-1}}(x,\a)\big|=\frac{1}{(x^2-\a)^{3/2}}\,\big| (-1)^n(\sqrt{-\a})^{2n+3}  \int_0^1    t^{2n+1}g_{R_{2n-1}}(t^2\a) \, \sqrt{1-t^2} \, dt\\
+ \int_{0}^x\,s^{2n+1}g_{R_{2n-1}}(s^2)\sqrt{s^2-\a}\, ds\,  \big| \\
=\frac{1}{(x^2+a^2)^{3/2}}\,\big| (-1)^na^{2n+3}  \int_0^1    t^{2n+1}g_{R_{2n-1}}(t^2a^2) \, \sqrt{1-t^2} \, dt
+ \int_{0}^x\,s^{2n+1}g_{R_{2n-1}}(s^2)\sqrt{s^2+a^2}\, ds     \, \big| \\
\lesssim \frac{1}{(x^2+a^2)^{3/2}}\,\left( a^{2n+3}  
+ \int_{0}^{|x|}\,s^{2n+1}\sqrt{s^2+a^2}\, ds     \right)\\
= \frac{1}{(x^2+a^2)^{3/2}}\,\left( a^{2n+3}  +x^{2n} \int_{0}^{|x|}\,s\sqrt{s^2+a^2}\, ds   \right) 
  \lesssim a^{2n} \frac{ a^{3} }{(x^2+a^2)^{3/2}} + |x|^{2n}\\
=  a^{2n}+ |x|^{2n}=|\a|^n+|x|^{2n}
\end{multline}

\medskip

{\bf For $\a=0$:} 

\smallskip

The  operator is given by
$\mathcal{J}_0f=-x^{-3}\int_0^x sf(s,0)\, ds$
and since $\gamma(0)=1$ and $\omega(x)=1+x\tilde{\omega}(x)$ we have
$$\mathcal{J}_0[\omega-\gamma^{-1}]=-x^{-3}\int_0^x \, s^2\tilde{\omega}(s)\,ds$$

Consider the Taylor approximation $\tilde{\omega}=P_n+R_n$ 
where $P_n(x)=\sum_{k=0}^{n-1}\omega_{k+1} x^{k}$ and $|R_n(x)|<C_n|x|^{n+1}$. Then $\mathcal{J}_0[\omega-\gamma^{-1}]=\mathcal{J}_0P_n+\mathcal{J}_0R_n$ where
$$\mathcal{J}_0P_n(x)=\sum_{k=0}^{n-1}\frac{\omega_{k+1}}{k+3}\,{x^k}$$
and 
$$|\mathcal{J}_0R_n(x)|\Le |x|^{-3}\int_0^{|x|}s^{n+3}C_n\, ds\Le C_{n} \, |x|^{n+1}$$

The proof of Lemma\,\ref{L2C} is now complete, and so is that of Proposition\,\ref{Lemma1}. $\Box$

\subsection{Continuation of $\xi(y)$ to $\RR$}\label{xicont}

$\xi(y):=\xi(x(y))$ is a solution of class $C^v$ for $y\in [-\delta_2',\delta_2]$ where $\sqrt{\beta}<\delta_2',\delta_2,$ and outside this interval there are no turning points.

\medskip

{\bf Remark.} The solution given by Proposition \ref{Lemma1} is invertible for small $x$: we can write \eqref{eqyx} as $$\sqrt{y^2-\beta}\,\frac{dy}{dx}=\omega(x)\sqrt{x^2-\alpha}.$$ By \eqref{eq:eqoY} and since $\omega>0$ we have $dy/dx>0$ for $$x\in(-\delta,\delta)\backslash \{\pm\sqrt{\a}\}$$ and $dy/dx\ne0$ at $x=\pm\sqrt{\a}$ since $$d\nv /dx\Big|_{x=\pm\sqrt{\a}}=1\mp2\sqrt{\a}w\ne0$$
Also we know that $x(y)$ is $C^v([-\delta_2',\delta_2])$ for any $\sqrt{\beta}<\delta_2'<y(-\sqrt{\a}+)$ and $\sqrt{\beta}<\delta_2<y(\sqrt{\alpha}-)$ by the inverse function theorem.

\medskip

Since $V(0)=1$ is the unique absolute max of $V$ then $|E-V(\xi)|>\delta_4$ for all $E$ with $|1-E|<{\delta_6}$ and $|\xi|>\delta_5$.

Let $[y_0,y_1)$ be an interval for which the solution $\xi(y)$ of \eqref{eqxi} is defined; we know that such an interval exists and, in fact we can choose $y_1\Ge\delta_2$ and $y_0\in(\sqrt{\beta},\delta_2)$, so that the equation {\eqref{eqxi}} can be normalized as
\bel{eqxi2}
 \frac{d\xi}{dy}=\sqrt{{y^2-\beta}}/ \sqrt{{E-V(\xi)}}
\ee
Since $\xi'(y)>0$ on $[y_0,y_1)$ then $\xi(y)>\xi_0{(:=\xi(y_0))}$ and by letting $\xi_0=\delta_5$, the right-hand side of \eqref{eqxi2} is continuous and bounded for $y\in[y_0,y_1)$ and $\xi\in [\xi_0,\xi(y_1)]$. Therefore $\xi(y_1-0)$ exists and $\xi(y)$ is a solution on $[y_0,y_1]$.

Since $|E-V(\xi(y_1))|\Ge\delta_4$ the solution can be continued beyond $y_1$. This shows that the maximal interval of existence of the solution $\xi(y)$ cannot be bounded, and the solution can be continued for all $y>y_0$.

Global existence for $y<0$ is similar.

Since $\xi'(y)\neq 0$ (by \eqref{eqxi}) then $y\to \xi(y)$ is one-to-one. \Bx

\subsection{Asymptotic behavior}\label{par:asyxi} Formula \eqref{asyxi} is obtained by direct asymptotic analysis on the differential equation \eqref{eqxi2}.

{\bf Remark.} The connection constants $C_\pm$ cannot be determined by this analysis, but they can be linked to $E-V(\xi)$ as follows.

{\em I. For $E<1$ (therefore $\beta>0$)} we have $\xi(\sqrt{\beta})=b, \ \xi(-\sqrt{\beta})=a$. Then, with the notation \eqref{defIp}
we have
\begin{multline}\label{Ipbp}
I_+(E)=
\int_{\sqrt{\beta}}^{+\infty}\left(\sqrt{y^2-\beta}-\sqrt{E}\xi'(y)\right)dy-b\sqrt{E}\\
=\lim_{y\to\infty} \left(\frac12 y\sqrt{y^2-\beta}-\frac 12\beta\ln(y+\sqrt{y^2-\beta})-\sqrt{E}\xi(y)\right)-\left( -\frac 12\beta\ln\sqrt{\beta}-\sqrt{E} \xi(\sqrt{\beta})  \right)-b\sqrt{E}\\ 
=-\frac 14\beta+\frac 14\beta\ln\beta-\frac 12 \beta\ln 2-\sqrt{E}C_+
\end{multline}
where we used \eqref{asyxi}. Similarly, using \eqref{defIm}
we have
\begin{multline}\label{Imbp}
I_-(E)
=-\lim_{y\to-\infty} \left(\frac12 y\sqrt{y^2-\beta}-\frac 12\beta\ln(-y-\sqrt{y^2-\beta})-\sqrt{E}\xi(y)\right)\\
+\left( \frac 12\beta\ln\sqrt{\beta}-\sqrt{E} \xi(-\sqrt{\beta})  \right)+a\sqrt{E}\\
=-\frac 14\beta+\frac 14\beta\ln\beta-\frac 12 \beta\ln 2 {-}\sqrt{E}C_-
\end{multline}

{\em II. For $E=1$ (therefore $\beta=0$)} we have $\xi(0)=0$, therefore $a=b=0$ and a calculation similar to the above yields \eqref{Ipbp}, \eqref{Imbp} in the limit $\beta\to 0+$.

{\em III. For $E>1$ (therefore $\beta<0$)} we define
\bel{defIpbn}
I_+(E):=\int_0^{+\infty}\left(\sqrt{E-V(\xi)}-\sqrt{E}\right)d\xi,\ \ I_-(E):=\int_{-\infty}^0\left(\sqrt{E-V(\xi)}-\sqrt{E}\right)d\xi
\ee
 We proceed as in the case $\beta>0$ (noting that $\xi=0$ when $x=0$):
 \begin{multline}
I_+(E)=\lim_{\xi\to \infty}\left[ \int_0^{\xi}\sqrt{E-V(\xi)}\, d\xi-\sqrt{E}\xi\right]=\lim_{y\to \infty}\left[ \int_{y(0)}^{y}\sqrt{t^2-\beta}\, dt-\sqrt{E}\xi \right]\\
=\lim_{y\to \infty}\left[ \int_{0}^{y}\sqrt{t^2-\beta}\, dt-\sqrt{E}\xi\right]-\int_{0}^{y(0)}\sqrt{t^2-\beta}\, dt\\
=\lim_{y\to\infty}\left[  \frac12 y\sqrt{y^2-\beta}-\frac 12\beta\ln(y+\sqrt{y^2-\beta})-\frac 14\beta\ln(-\beta)
-\sqrt{E}\xi \right]-\int_{0}^{y(0)}\sqrt{t^2-\beta}\, dt\\ \label{eq57}
=-\frac 14\beta+\frac 14\beta\ln(-\beta)-\frac 12 \beta\ln 2-\sqrt{E}C_+-\int_{0}^{y(0)}\sqrt{t^2-\beta}\, dt
\end{multline}
 (where we used \eqref{asyxi}).
To evaluate the last integral note that on one hand we have (recall the notations $y=\sqrt{\beta/\alpha}Y$, $Y=x+(\alpha-x^2)w$)
\bel{oint111} 
\int_{0}^{y(0)}\sqrt{t^2-\beta}\, dt={\frac{\beta}{\alpha}}\int_0^{Y(0)}\sqrt{s^2-\a}\,ds = -\beta\int_0^{\sqrt{-\alpha}w(0)}\sqrt{\tau^2+1}\,d\tau=-\frac{\beta}{\a}\phi_\omega(\a)
\ee
where the last equality is obtained by noting that, on the other hand, equation \eqref{eq:eq4} at $x=0$ simplifies to
$$w(0)=\gamma \, u(0;\alpha)+w^2\int_0^1\frac{(1-\sigma)\a w(0)\sigma}{\sqrt{1-\a w(0)^2\sigma^2 }}\, d\sigma $$
which, using \eqref{defumR}, and, after integration by parts, becomes
\bel{oint222}
\frac{1}{\a}\phi_\om(\a)=\int_0^{\sqrt{-\alpha}w(0)}\sqrt{\tau^2+1}\,d\tau
\ee

Using \eqref{oint111} relation \eqref{eq57} becomes
\bel{Ipbn}
I_+(E)=-\frac 14\beta+\frac 14\beta\ln(-\beta)-\frac 12 \beta\ln 2-\sqrt{E}C_+ + \gamma^{-1}\phi_\om(\a)
\ee
Similar calculations give
\bel{Ipbn*}
I_-(E)=-\frac 14\beta+\frac 14\beta\ln(-\beta)-\frac 12 \beta\ln 2-\sqrt{E}C_- - \gamma^{-1}\phi_\om(\a)
\ee

Note that in all cases we have
\bel{cppcm}
-\sqrt{E}(C_++C_-)=I_+(E)+I_-(E)+\frac 12\beta-\frac 12\beta\ln{|\beta|}+\beta\ln2
\ee

On one hand, $\gamma^{-1}\phi_\omega$ has an explicit construction. On the other hand, note that we can rewrite the integral in \eqref{oint111} as
$$- \gamma^{-1}\phi_\om(\a)=\int_0^{y(0)}\sqrt{y^2-\beta}\, dy=\int_{\xi_0}^0\sqrt{E-V(\xi)}\, d\xi$$
where $\xi_0$ is the value of $\xi$ for $y=0$ (recall that $\xi(y(0))=0$). In the $C^\infty$ case of course, only the Taylor coefficients of all quantities at $E=1$ are relevant.

\subsection{$\xi(y)$ behaves like a symbol}\label{symbxi}
Assume that $|V^{(k)}(\xi)|\lesssim \la \xi\ra^{-c-k}$ (for some $c>1$) for all integer $k\Ge 0$.

The fact that $\xi'(y)\sim \pm \frac{y}{\sqrt{E}}$ is obtained by direct asymptotic analysis on the differential equation \eqref{eqxi2}. Next, differentiating the equation we obtain
$$\xi''(y)=\frac{y}{\sqrt{{y^2-\beta}}\,\sqrt{{E-V(\xi)}}}{+\frac{V'(\xi)}{2({E-V(\xi)})} \xi'(y)^2}      $$
and since $|V'(\xi(y))|\lesssim \la y^2\ra^{-c-1}$ then $|\xi''(y)|\lesssim 1$. The other derivatives are proved by induction on $k$.

\section{Scattering theory of \eqref{eqf}:  $\hb/\beta\Le {\rm Const.}$ and $\beta>0$}\label{Aicase}

We now apply the change of variables of the previous section to the problem of obtaining 
fundamental systems of the semi-classical \eqref{eqf} with precise control of the asymptotic behavior
both in terms of small $\beta$ and small $\hbar$. 
To fix the turning points we substitute (with $f$ as in \eqref{eq:S})
\bel{psi12}
h_1=\hb/\beta,\ \ y=\sqrt{\beta}y_1,\ \ {\psi_2}(y):=\psi_1(y_1),\ \ \beta f(\sqrt{\beta}y_1):=f_1(y_1)
\ee
which transforms \eqref{eqpsi2} to
\bel{eqpsi1}
\frac{d^2\psi_1}{dy_1^2}=h_1^{-2}(1-y_1^2)\psi_1+f_1\psi_1
\ee
which can be viewed as a perturbation of the Weber equation
\bel{eqw}
\frac{d^2w}{dy_1^2}=h_1^{-2}(1-y_1^2)w
\ee

Inspired by the main terms of the asymptotics in \cite{Olver1975}, \cite{Olver1959} (only those results present the error in additive form, and we need multiplicative) we proceed as follows.
We denote, see \cite{Olver1975}, 
\bel{eta}
\eta(y)=\left(\frac 32\int_1^y\sqrt{t^2-1}\, dt\right)^{2/3}\ {\rm for \ }y\Ge 1,\ \ \eta(y)=-\left(\frac 32\int_y^1\sqrt{1-t^2}\, dt\right)^{2/3}\ {\rm for \ }y{\, \in[0,1]} 
\ee
and let
\bel{defg}
g(y)=\left(\frac{\eta(y)}{y^2-1}\right)^{1/4}
\ee

Denote

\bel{defAB}
A(y_1)=g(y_1){\rm Ai}\left(-h_1^{-2/3}\eta(y_1)\right),\ \ B(y_1)=g(y_1){\rm Bi}\left(-h_1^{-2/3}\eta(y_1)\right)
\ee
where ${\rm Ai},\,{\rm Bi}$ are the Airy functions. These functions will be now used to construct fundamental systems
of \eqref{eqpsi1}, \eqref{eqw}, respectively.

\subsection{The exponential region: $y_1\in [0,1]$}\label{S1}

\bl\label{P1AB} 

(i) For $y_1\in[0,1]$ eq. \eqref{eqpsi1} has two independent solutions of the form
\bel{solp1A}
\psi_{1,A}(y_1)=A(y_1)\left(1+ h_1\,a_1(y_1;h_1,\beta)\right)
\ee
\bel{solp1B}
\psi_{1,B}(y_1)=B(y_1)\left(1+ h_1\,b_1(y_1;h_1,\beta)\right)
\ee
where the error terms $a_1,b_1$ satisfy for all $k,l\Ge0$, with $\lam=h_1^{2/3}$,
\bel{estimab}
\begin{array}{ll}
\big|\partial_{y_1}^k\partial_{\beta}^l a_1\big|\Le \,C_{kl}\,  (-\eta(y_1))^{\frac12 -k}\,\beta^{-l},\ \big|\partial_{y_1}^k\partial_{\beta}^l b_1\big|\Le \,C_{kl}\,  (-\eta(y_1))^{\frac12 -k}\,\beta^{-l} & \text{if }-\eta(y_1)\Ge h_1^{2/3}\\
\big|\partial_{y_1}^k\partial_{\beta}^l a_1\big|\Le \,C_{kl}\,  h_1^{\frac13-\frac{2k}3}\,\beta^{-l},\ \big|\partial_{y_1}^k\partial_{\beta}^l b_1\big|\Le \,C_{kl}\,  h_1^{\frac13-\frac{2k}3}\,\beta^{-l} & \text{if }-\eta(y_1)\in[0, h_1^{2/3})
\end{array}
\ee
(note that $\beta^{-l}< C_{l}\, \hb^{-l}$).

In particular, at $y_1=1$:
\bel{estab1}
\begin{array}{l}\big|\partial_{\beta}^l a_1(1,h_1,\beta)\big|\Le h_1^{1/3}\,C_{l}\,\beta^{-l},\ \ \ \big|\partial_{\beta}^l b_1(1,h_1,\beta)\big|\Le \, h_1^{1/3}\,C_{l}\, \beta^{-l}\\
\big|\partial_{y_1}\partial_{\beta}^l a_1(1,h_1,\beta)\big|\Le h_1^{-1/3}\,C_{l}\,\beta^{-l},\ \ \ \big|\partial_{y_1}\partial_{\beta}^l b_1(1,h_1,\beta)\big|\Le \, h_1^{-1/3}\,C_{l}\, \beta^{-l}
\end{array}
\ee

(ii) Furthermore, for $y_1\in[0,1]$
$$\psi_{1,A}'(y_1)=A'(y_1) (1+ h_1\, {a^d_1(y_1;h_1,\beta)} )$$
$$\psi_{1,B}'(y_1)=B'(y_1) (1+ h_1\, {b^d_1(y_1;h_1,\beta)} )$$
where the error terms ${a^d_1,b^d_1}$ satisfy estimates similar to \eqref{estab1} uniformly in $y_1\in[0,1]$.  
\el

The {proof} is found in \S\ref{PfL3}.

\medskip

In particular, the Weber equation \eqref{eqw} also admits a fundamental system $w_A,w_B$ approximated as in Lemma\,\ref{P1AB}, therefore
we have: 

\begin{Corollary}\label{Cor4}
$$\psi_{1,A}(y_{1})=w_A (1+h_1 \tilde{a}_1(y_1;h_1,\beta) )$$
$$\psi_{1,B}(y_{1})=w_B (1+h_1 \tilde{b}_1(y_1;h_1,\beta) )$$
with $\tilde{a}_1,\tilde{b}_1$ satisfying \eqref{estimab},\,\eqref{estab1}.

Furthermore, 
$$\psi'_{1,A}=w'_A  (1+h_1 {\tilde{a}^d_1(y_1;h_1,\beta)} )$$
$$\psi'_{1,B}=w'_B  (1+h_1 {\tilde{b}^d_1(y_1;h_1,\beta)} )$$
with ${\tilde{a}^d_1,\tilde{b}^d_1}$ satisfying \eqref{estimab},\,\eqref{estab1}.
\end{Corollary}

\subsection{The oscillatory region: $y_1\Ge 1$}\label{S2} 

The fundamental systems which we just constructed for $0\Le y_{1}\Le 1$ extends to $y_{1}\Ge1$. 
In order to determine the asymptotic behavior of these solutions as $y_{1}\to\infty$, we construct
a new fundamental system in that regime, namely 
the  {\em Jost solutions}. This standard terminology refers to oscillatory solutions which asymptotically
equal those of the free problem, i.e., $e^{\pm i y_{1}\lambda}$. See for example~\cite[Section 1.3]{CDST}. 
Note carefully, though, that we are using a {\em global change of variables} in \eqref{eqxi} which reduces matters
not to the free problem but to the (global) Weber equation. This leads to  different  asymptotic behavior, 
as given by the following lemma.  

\bl\label{WtoAB} We have 
$$A(y_1)\pm iB(y_1)=g(y_1)({\rm Ai}\pm i{\rm Bi})(-h_1^{-2/3}\eta(y_1))\sim \, \lambda_\pm \ y_1^{-\frac 12\pm\frac{ i}{2h_1}}\,e^{\mp iy_1^2/(2h_1)}\ \ \ (y_1\to +\infty)$$
where
$$\lambda_+=\overline{\lambda_-}=\pi^{-1/2}h_1^{1/6}e^{i\pi/4}(4e)^{i/4h_1}$$
\el

The {proof} is found in \S\ref{PfL5}.

\bl\label{L5}  The Jost solutions of equation \eqref{eqpsi1} are as follows: 
\begin{enumerate}
\item[(i)] For $y_1\Ge 1$ \eqref{eqpsi1} has two independent solutions of the form
\bel{solppm}
\psi_{1,\pm}(y_1)=\left[A(y_1)\mp i B(y_1)\right]\,  \left(1+h_1 c_\pm(y_1;h_1,\beta)\right)
\ee
where
\bel{estimc}
\big|\partial_{y_1}^k\partial_{\beta}^l c_\pm\big|\Le C_{kl}\, \la y_1\ra^{-2-k}\beta^{-l}
\ee
Also, at $y_1=1$:
\bel{estimc1}
\big|\partial_{\beta}^l c_\pm(1,h_1,\beta)\big|\Le \,C_{l}\,\beta^{-l},\ \ \big|\partial_{\beta}^l \partial_{y_1}\, c_\pm(1,h_1,\beta)\big|\Le h_1^{-2/3}\,C_{l}\,\beta^{-l}
\ee
\item[(ii)] Furthermore:
$$\psi'_{1,\pm}(y_1)=\left[A'(y_1)\mp i B'(y_1)\right]\,  \left(1+h_1 {c^d_\pm(y_1;h_1,\beta)}\right)$$
where the error terms ${c^d_\pm}$ satisfy estimates similar to \eqref{estimc},\,\eqref{estimc1} for $y_1\Ge 1$.
\end{enumerate}
\el

The proof is found in \S\ref{PFLem6}.
In particular, for $f\equiv 0$, we obtain Jost solutions of the unperturbed semi-classical Weber equation.

\begin{Corollary}\label{C6}
 The Weber equation \eqref{eqw} has two solutions $w_\pm$ estimated as in Lemma~\ref{L5}.
\end{Corollary}

From Lemma~\ref{L5} and {Corollary}\,\ref{C6} it follows that

\begin{Corollary}\label{C7}
$$\psi_{1,\pm}=w_\pm\left(1+h_1 \tilde{c}_\pm(y_1;h_1,\beta)\right)$$
with $\tilde{c}_\pm$ satisfying \eqref{estimc},\eqref{estimc1}.
\end{Corollary}

\subsection{The regions with  $y_1\Le 0$}\label{y1neg}

Changing  variables $$y_1=-y_3, \quad \psi_{1}(y_1)=\psi_{1}(-y_4):={\psi_4}(y_4)$$ we see that 
if $\psi_{1}(y_1)$ solves \eqref{eqpsi1} then ${\psi_{4}}(y_4)$ solves an equation with the same properties as \eqref{eqpsi1}, therefore the results of \S\ref{S1} and \S\ref{S2} apply to ${\psi_{4}}(y_4)$ for $y_4\Ge 0$. Reverting to the original variable $y_1$ we obtain:

\bl 
\begin{enumerate}
\item[(i)] For $y_1\in[-1,0]$ equation \eqref{eqpsi1} has two independent solutions of the form
$$\psi_{1,A}^\ell(y_1)=A(-y_1) (1+h_1 a^\ell_1(y_1;h_1,\beta) )$$
$$\psi_{1,B}^\ell(y_1)=B(-y_1) (1+h_1 b^\ell_1(y_1;h_1,\beta) )$$
where the error terms $a^\ell_1,b^\ell_1$ satisfy \eqref{estimab},\eqref{estab1}.

\item[(ii)] In particular, the Weber equation \eqref{eqw} also has two solutions $w^\ell_A,w^\ell_B$ of the form in (i) for $y_1\in[-1,0]$ and thus
$$\psi^\ell_{1,A}=w^\ell_A\left(1+h_1 \tilde{a}^\ell_1(y_1;h_1,\beta)\right)$$
$$\psi^\ell_{1,B}=w^\ell_B\left(1+h_1 \tilde{b}^\ell_1(y_1;h_1,\beta)\right)$$
with $\tilde{a}^\ell_1,\tilde{b}^\ell_1$ satisfying \eqref{estimab},\eqref{estab1}.

\item[(iii)] For $y_1\Le -1$ eq. \eqref{eqpsi1} has two independent solutions of the form
$$\psi^\ell_{1,\pm}(y_1)=\left[A(-y_1)\mp i B(-y_1)\right]\,  \left(1+h_1 c^\ell_\pm(y_1;h_1,\beta)\right)$$
where $c^\ell_\pm$ satisfy \eqref{estimc},\eqref{estimc1}.

\item[(iv)] In particular, the Weber equation \eqref{eqw} also has two solutions $w^\ell_\pm$ of the form in (iii), therefore
$$\psi^\ell_{1,\pm}=w^\ell_\pm\left(1+h_1 \tilde{c}^\ell_\pm(y_1;h_1,\beta)\right)$$
with $\tilde{c}^\ell_\pm$ satisfying \eqref{estimc},\eqref{estimc1}.
\end{enumerate}
\el

\subsection{Matching at $y_1=\pm 1$} This is, as expected, straightforward (in \cite{Olver1975}, \cite{Olver1959}, the Airy asymptotic approximation used is valid from infinity, through one turning point and past $y_1=0$). We use the notation $[f\; g]$ to denote the row vector
with functions $f,g$. 

\bl\label{L10} Matching at $y_1=1$:
denote
\bel{notP}
\Psi_{\pm}=[\psi_{1,+}\ \psi_{1,-}],\ \Psi_{AB}=[\psi_{1,A}\ \psi_{1,B}],\  W_\pm=[w_{+}\ w_{-}], \ W_{AB}=[w_A\ w_B]
\ee
We have
$$\Psi_{\pm}= \Psi_{AB}\, (E_0+h_1E_1(h_1))$$
and, as a consequence
$$W_\pm=\, W_{AB}(E_0+h_1E_2(h_1))$$
where
$$E_0=\left[\begin{array}{cc} 1& 1\\-i & i\end{array}\right]$$
and $E_{1,2}(h_1)$ are square matrices with bounded entries (for $h_1\lesssim 1$, 
and $\beta<\beta_0$).

Similar results hold at $y=-1$.
\el

See \S\ref{PfL9} for the proof. Here $\beta_{0}>0$ is small so that the results of the previous section apply.

\subsection{Matching at $y_1=0$}\label{Match0} This is equivalent to finding the monodromy of equation \eqref{eqpsi1} which is estimated based on the monodromy of the modified parabolic cylinder functions (see Appendix B) as follows.

\bp\label{Prop13}
We have 
\bel{MonP}
\Psi_\pm^\ell=\Psi_\pm N,\ \ {\rm with\ } N=(I+h_1R)M(I+h_1T)
\ee 
where $M$ is the monodromy matrix of the Weber equation \eqref{eqw} given by \eqref{monoWM} and the matrices $R,T$ have bounded entries in the parameters for $h_1{\lesssim 1}, 
\beta<\beta_0$.
\ep

\begin{proof}
We use the notations \eqref{notP} and similar ones for $y_1\Le0$: $\Psi_{\pm}^\ell=[\psi_{1,+}^\ell \psi_{1,-}^\ell]$ etc. The following table summarizes the ranges of validity of the different fundamental sets of solutions used:
\begin{center}
\begin{tabular}{c|lcccccccr}
\hline
$y_1$ & $-\infty$ & \ & $-1$ &\  & $0$ & \ &  $1$ & \ & $+\infty$ \\
\hline
 & &\\
solutions of \eqref{eqpsi1} & \ & $\Psi^\ell_{\pm} $ & \ & $\Psi^\ell_{AB}$ & \  & $\Psi_{AB}$ & \ & $\Psi_{\pm}$\\
\hline
 & &\\
solutions of \eqref{eqw} & \ & $W_\pm^\ell $  & \ & $W^\ell_{AB}$ & \  & $W_{AB}$ & \ & $W_{\pm}$\\
\hline
 & &\\
approx. solutions  & \ & $(A\mp iB)(-y_1)$  & \ & ${A,B(-y_1)}$ & \  & ${A,B}$ & \ & $A\mp iB$\\
\hline
 \hline
\end{tabular}
\end{center}

\medskip

Combining the relations:
\begin{enumerate}
\item[(0)] $W_\pm^\ell=W_{\pm}M$, see \eqref{monoW}, \eqref{monoWM},  
\item[(1)] $\Psi_\pm=W_\pm(I+h_1D(y_1))$ where $D_1$ is diagonal, see Corollary\,\ref{C7}, 
\item[(2)] $\Psi_\pm=\Psi_{AB}(E_0+h_1E_1)$, see Lemma\,\ref{L10}, 
\item[(2')] $W_\pm=W_{AB}(E_0+h_1E_2)$, see Lemma\,\ref{L10}, 
\item[(3)] $\Psi_{AB}=W_{AB}(I+h_1F(y_1))$, see {Corollary}\,\ref{Cor4},  
\end{enumerate}
and similarly  
\begin{align*}
\Psi_\pm^\ell &=W_\pm^\ell(I+h_1D^\ell(y_1)) \\
\Psi_\pm^\ell &=\Psi_{AB}^\ell(E_0+h_1E_1^\ell) \\
 W_\pm^\ell &=W_{AB}^\ell(E_0+h_1E_2^\ell)
 \end{align*}
we obtain
\begin{align*}
\Psi_\pm(0) &=\Psi_{AB}(0)(E_0+h_1E_1)=W_{AB}(0)(I+h_1F(0))(E_0+h_1E_1)\\
&=W_\pm(0)(E_0+h_1E_2)^{-1}(I+h_1F(0))(E_0+h_1E_1)
\end{align*} 
and a similar expression for $\Psi_\pm^\ell(0)$, which implies that $\Psi_\pm^\ell(0)=\Psi_\pm(0) N$
for 
$$
N=(E_0+h_1E_1)^{-1}(I+h_1F(0))^{-1}(E_0+h_1E_2)M(E_0+h_1E_2^\ell)^{-1}(I+h_1F(0))(E_0+h_1E_1^\ell)$$ which has the stated form. 
\end{proof}

\subsection{The scattering matrix}\label{ScatMat}

Equation \eqref{eqf} has Jost solutions $f_\pm^{\ell,r}$ (since $V\in L^1(\RR)$) and it is easy to see that they have the asymptotic behavior
$$ f_\pm^\ell(\xi) = e^{\pm i\frac{\sqrt{E}}{\hb}\xi} \left(1+o(1)\right),\ \ (\xi\to -\infty)$$
$$f_\pm^r(\xi) = e^{\pm i\frac{\sqrt{E}}{\hb}\xi} \left( 1+o(1)\right),\ \ (\xi\to +\infty)$$
(if, say, $V(\xi)\sim B \xi^{-r-1}$ with $r>0$, then the correction $o(1)$ is $O(\xi^{-r})$) and using \eqref{asyxi}, we obtain

\bel{fpmr}
f_\pm^r(\xi(y)) = { e^{\pm i\frac{\sqrt{E}C_{+}}{\hb}}}y^{\mp\frac{i\beta}{2\hb}}e^{\pm\frac{iy^2}{2\hb}}(1+o(1))\ \ (y\to +\infty)
\ee

\bel{fpml}
f_\pm^\ell(\xi(y)) = { e^{\mp i\frac{\sqrt{E}C_{-}}{\hb}}}  (-y)^{\pm\frac{i\beta}{2\hb}}e^{\mp\frac{iy^2}{2\hb}}(1+o(1))\ \ (y\to -\infty)
\ee

{On the other hand, we work} back through the substitutions $f(\xi(y))=\sqrt{\xi'(y)}{\psi_2}(y)$, followed by \eqref{psi12}.  We have, from \eqref{asyxi}, that 
\bel{asydxi}
\sqrt{\xi'(y)}=y^{1/2}E^{-1/4}(1+O(y^{-2}))\ \ \ (y\to \infty)
\ee
 and using Corollary\,\ref{C7} we see that 
$$ f_{\pm}^r(\xi(y))=K_\pm\, \sqrt{\xi'(y)}\,  \psi_{1,\pm}(y\sqrt{\beta})=K_\pm\, \sqrt{\xi'(y)}\,  w_{\pm}(y\sqrt{\beta})(1+O(y^{-2})),$$
where $$K_\pm= \frac{ E^{1/4} e^{\pm i\frac{\sqrt{E}C_{+}}{\hb}}}{\lambda_\mp\beta^{\frac 14\pm\frac{i\beta}{4\hb}}}.$$
Similarly, 
$$ f_{\pm}^\ell(\xi(y))=K_\mp^\ell\, \sqrt{\xi'(y)}\,  \psi_{1,\mp}^\ell(y\sqrt{\beta}) ,\ \text{with }K_\pm^\ell= \frac{ E^{1/4} e^{\pm i\frac{\sqrt{E}C_{-}}{\hb}}}{\lambda_\mp\beta^{\frac 14\pm\frac{i\beta}{4\hb}}}$$
We now use Proposition\,\ref{Prop13}:
\begin{align*}
\frac 1{\sqrt{\xi'(y)}}F^{\ell} &=\Psi_\pm^{\ell}{\left[\begin{array}{cc} 0 & K_+^{\ell}\\ K_-^{\ell} & 0\end{array}\right]} = \frac 1{\sqrt{\xi'(y)}}F^r {\left[\begin{array}{cc} 1/K_+ & 0\\ 0 & 1/K_- \end{array}\right] \,N\, \left[\begin{array}{cc} 0 & K_+^{\ell}\\ K_-^{\ell}& 0\end{array}\right]}\\
&=  \frac 1{\sqrt{\xi'(y)}}F^r\mathcal{M}  
\end{align*}
where 
$F^{r/\ell}:=[f_+^{r/\ell}\, f_-^{r/\ell}]$ and $\mathcal{M}=(I+h_1\mathcal{R})\mathcal{M}_0(I+h_1\mathcal{T})$ 

with $\mathcal{M}_0$ obtained by a straightforward calculation as
\bel{monof}
\mathcal{M}_0=
\left[\begin{array}{cc} {p}e^{i\phi}\sqrt{1+A^2} &-{q^{-1}}iA\\
i{q}A & {p^{-1}}e^{-i\phi}\sqrt{1+A^2} \end{array}\right]
\ee
where
\bel{valsA}
A=e^{\pi/(2h_1)},\ \ e^{i\phi}=e^{i\phi_2}\left(\beta h_1/2\right)^{i/(2h_1)}=e^{i\phi_2}\left(\hb/2\right)^{i\beta/(2\hb)},\ \ \ 
\phi_2=\arg\Gamma\left(\frac 12+\frac{i}{2h_1}\right)
\ee
and $$
{p=e^{ -i\frac{\sqrt{E}}{\hb}(C_{-}+C_+)}, \ \ q=e^{ -i\frac{\sqrt{E}}{\hb}(C_{-}-C_+)}}$$

Note that the entries $\mathcal{M}_{ij}$ of the monodromy matrix $\mathcal{M}$ are linked to the entries $\mathcal{M}_{0,ij}$ of $\mathcal{M}_0$ by $\mathcal{M}_{ij}=\mathcal{M}_{0,ij}(1+h_1\mathcal{P}_{ij})$ where $\mathcal{P}_{ij}$ is multilinear in the entries of $\mathcal{R}$, $\mathcal{T}$ and bounded in the parameters.

The entries $\mathcal{S}_{ij}$ of the scattering matrix $\mathcal{S}$ can now be calculated as
$${\mathcal{S}_{11}}=\frac{\det\mathcal{M}}{\mathcal{M}_{22}}={e^{i\phi -i\frac{\sqrt{E}}{\hb}(C_{-}+C_+)}}\frac{1}{\sqrt{1+A^2}}\,(1+h_1e_{21})$$
and
$${{\mathcal{S}_{12}}}=-\frac{\mathcal{M}_{21}}{\mathcal{M}_{22}}=-ie^{i\phi{-i\frac{2\sqrt{E}C_-}{\hb}}}\frac{A}{\sqrt{1+A^2}}\,(1+h_1e_{11}).$$
(with notations as in \eqref{valsA}).
{Using \eqref{Ipbp}, \eqref{Imbp},}
 \eqref{valsA} we obtain 
{\bel{argS11}
\phi -\frac{\sqrt{E}}{\hb}(C_{-}+C_+)=\frac{1}{\hb}(I_+(E)+I_-(E))+\phi_2+\frac{\beta}{2\hb}[1+\ln(2\hb/\beta)]
\ee
where $I_+(E),I_-(E)$ are defined in \eqref{defIp}, \eqref{defIm}.}

Therefore
\bel{S11}
{\mathcal{S}_{11}}=e^{\frac i{\hb}(I_+(E)+I_-(E))}\, e^{i \theta}\frac{1}{\sqrt{1+A^2}}\,(1+\frac{\beta}{\hb}e_{11})
\ee
where
\bel{formtheta}
\theta=\phi_2+\frac{1}{2h_1}\left[1+\ln(2h_1)\right]=\phi_2+\frac{\beta}{2\hb}\left[1+\ln(2\hb/\beta)\right]
\ee
Similarly, {using \eqref{Imbp},} \eqref{valsA} we obtain
{$$\phi -2\frac{\sqrt{E}}{\hb}C_-=\frac{2}{\hb}I_-(E)+\phi_2+\frac{\beta}{2\hb}\left[ 1+\ln(2\hb/\beta)\right]]$$}
and therefore
\bel{S21}
{\mathcal{S}_{21}}={e^{\frac i{\hb}2I_-(E)}\, e^{i \theta}\frac{-iA}{\sqrt{1+A^2}}\,(1+h_1e_{21})}
\ee

\subsubsection{Dominant terms for small $h_1=\hb/\beta$}\label{smallh1} 
Using  $$\phi_2=(-1-\ln(2h_1))/(2h_1)+O(h_1)$$ (recall that $h_1=\hb/\beta$) we see that in \eqref{formtheta} we have $\theta=O(h_1)$. 
The modulus in ${\mathcal{S}_{11}}$ is, using \eqref{SdeE},
$$(1+A^2)^{-1/2}\sim A^{-1}=\exp[-\pi\beta/(2\hb)]=\exp(-S(E)/\hb)$$
Similarly, the modulus of ${{\mathcal{S}_{12}}}$ is of order $1$, while 
the argument in ${{\mathcal{S}_{12}}}$ is 
$$\phi -2\frac{\sqrt{E}}{\hb}C_{-}=\frac{2}{\hb}I_-(E)\,+O(h_1)$$
The dominant terms in these expressions correspond to the ones in \cite{Ramond}. 


\section{Proofs of statements in \S\ref{Aicase}}\label{pfP1AB} 

The proofs use lemmas found in \cite{CDST}. It is useful to note the following identities satisfied by the functions defined in \eqref{eta}, \eqref{defg}, \eqref{defAB}:

\bel{useide}
g^2\eta'=1,\ \ \eta\eta'^2=y^2-1,\ \ \eta/g^4=y^2-1
\ee
and note that $\eta(y)$ is $C^\infty$, increasing, with $\eta(0)\approx -1.11,\ \eta(1)=0$ and 
\bel{asyetag}
\eta(y)\sim (3/4)^{2/3}y^{4/3}, \ \ { g^4(y)\sim (3/4)^{2/3}\, y^{-2/3}  \ \ \ \text{for }y\to\infty}
\ee

\subsection{Reduction to the Airy equation}\label{AiryRed}

In equation \eqref{eqpsi1} we change the dependent and independent variables:
\bel{psi1calA}
\psi_1(y_1)=g(y_1)\mathcal{A}\left(-\lambda^{-1}\eta(y_1)\right)\ \ \ \text{where } \lambda=h_1^{2/3}
\ee
and let
$$x=-\lambda^{-1}\eta(y_1) \text{ with its inverse }y_1=\eta^{-1}(-\lam x):=\theta(\lam x)$$
Equation \eqref{eqpsi1} becomes
\bel{eqA}
\frac{d^2}{dx^2}\mathcal{A}=x\mathcal{A}+\lambda^2g^4V_1\mathcal{A},\ \ \ {\rm where\ }V_1=f_1-g''/g
\ee
which is a perturbation of the Airy equation. 
The following table summarizes ranges of different variables used:

\bigskip

\begin{tabular}{l|lcccccc}
$y_1$ & $0$ &\ & $1$ &\ \ \ \  \ \ & $y_c$ & \ \ \ & \ \  \ \ \ \ \ \  $+\infty$  \\
\hline
$\eta(y_1)$ & $\eta(0)\approx -1.1$ \ \ \ &\ \ \ \  $-$\ \ \ \ & $0$ &\  +\ & $1$ &\  \ $+$\ \ \ \ & \ \ \ \ \ \ \  $+\infty$  \\
\hline
$x$ & $-\lambda^{-1}\eta(0)$\ \ \  & $+$& $0$ &\ \ \ \  $-$\ \ \ & $-\lambda^{-1}$ & \ $-$ \ & \ \ \ \ \ \  \ $-\infty$  \\
\hline
\end{tabular}

\subsection{Proof of Lemma\,\ref{P1AB} }\label{PfL3}

\begin{Lemma}\label{estinx}

Equation \eqref{eqA} has solutions $\mathcal{A}_B, \mathcal{A}_B$ the form 
\bel{erforB}
\mathcal{A}_B(x)={\rm Bi}(x)(1+b(x;\lam))
\ee 
\bel{erforA}
\mathcal{A}_A(x)={\rm Ai}(x)(1+a(x;\lam))
\ee
where the errors satisfy
\bel{esbx}
|\partial_\lambda^\ell\partial_x^k b(x;\lam)|\Le C_{k,\ell}\la x\ra^{1/2-k}\lambda^{2-\ell},\ \ |\partial_\lambda^\ell\partial_x^k b(x;\lam)|\Le C_{k,\ell}\la x\ra^{1/2-k}\lambda^{2-\ell}
\ee
\end{Lemma}

{\em Proof.} 

Substituting \eqref{erforB} in \eqref{eqA}, the equation for the error $b(x;\lam)$ can be turned into the Volterra equation
\bel{eqforb}
b(x;\lam)=\int_0^x\, dx'\, K(x',x)[1+b(x';\lam)] 
\ee
with
$$K(x',x)=\lambda^2(g^4V_1)\big|_{y_1=\theta(\lam x')}\, {\rm Bi}^2(x')\int_{x'}^x
\frac{dx''}{{\rm Bi}^2(x'')}\ \ \text{ for }x>x'$$
Note that $V_1$ depends of $\beta$ through the term $f_1(y_1)=\beta f(y_1\sqrt{\beta})$ where $f(y)=f(y,\beta)$ and here we should consider $\beta=\beta(\lam)=\hb \lam^{-\frac32}$. 

Straightforward calculations show that for $0\Le x\Le -\lambda^{-1}\eta(0)$ we have
$$|\partial_\lam^\ell\partial_x^k \lambda^2(g^4V_1)\big|_{y_1=\theta(\lam x')}|\Le C_{k,\ell}\langle x\rangle^{-k}\lambda^{2-\ell}$$

Proposition C8 in \cite{CDST} can be applied (see its statement in \S\ref{PropC8}), yielding the fact that equation \eqref{eqforb} has a unique solution, and this solution satisfies \eqref{esbx}.

(b) An independent solution $\mathcal{A}_A$ of \eqref{eqA}
 is obtained using the fact that the Wronskian $[\mathcal{A}_A,\mathcal{A}_B]=Const.$ We choose this constant to be the value $W[{\rm Ai},{\rm Bi}]=\pi^{-1}$. This implies 
$\mathcal{A}_A'-\mathcal{A}_A\, {\mathcal{A}_B'}/{\mathcal{A}_B}=-{\pi^{-1}}/{\mathcal{A}_B}$
which, using \eqref{erforB}, we rewrite in integral form as
$$ \mathcal{A}_A(x)={\rm Bi}(x)(1+b(x;\lam))\int_{-\lambda^{-1}\eta(0)}^x\, dt\, \frac{-\pi^{-1}}{{\rm Bi}(t)^2(1+b(t;\lam))^2} $$
It is standard to check that this implies \eqref{erforA} where  $a(t;\lam)$ satisfies the same estimates as $b(t;\lam)$ does, namely \eqref{esbx}. 

(c) To show that the derivatives of these solutions are approximated by derivatives of the Airy functions, differentiating \eqref{erforB} with respect to $x$ we obtain
$$\mathcal{A}_B'={\rm Bi}'(1+b)+{\rm Bi}\,b'={\rm Bi}'\left( 1+b+b'\, {\rm Bi}/{\rm Bi}'\right)
$$
(${\rm Ai},{\rm Ai}',{\rm Bi},{\rm Bi}'$ have no zeroes for $x\Ge 0$ \cite{nistAi})

The asymptotic behavior at infinity of the Airy functions shows that $|{\rm Bi}/{\rm Bi'}|\lesssim \la x\ra^{-1/2}$ hence $|b'\, {\rm Bi}/{\rm Bi}'| \lesssim \la x\ra^{-1/2}\, \la x\ra^{\frac12-1}\lambda^2$ by (\ref{esbx}) { and therefore the error satisfied the same estimates as $b$ does. The estimates for $\mathcal{A}_A'$  are similar. } $\Box$.

\

 Lemma\,\ref{P1AB} follows from Lemma\,\ref{estinx} by simply going back to the original variables using \eqref{psi1calA}: we found a solution $$\psi_{1,B}(y_1)=g(y_1)\mathcal{A}_{B}(-\lambda^{-1}\eta(y_1))=B(y_1)(1+h_1b_1(y_1;h_1,\beta))$$ where $b_1(y_1;h_1,\beta)=h_1^{-1}b(-\lambda^{-1}\eta(y_1),\lam)$. The solution $\psi_{1,A}(y_1)$ is similar.

The estimates \eqref{esbx} are straightforwardly transferred to estimates for $b_1$. Clearly \eqref{esbx} implies that $|b_1|\lesssim h_1^{-1}\la -\eta/\lam\ra^{1/2}\lam^2$ which, for $ -\eta>\lam$ is of order $\sqrt{-\eta}<Const.$, and for $ -\eta\in[0,\lam)$ is $h_1^{1/3}$. Then {$$|\partial_{y_1}b_1|=h_1^{-1}| (-\eta'/\lam) \partial_{x}b|\lesssim h_1^{-1} \lam\la -\eta/\lam\ra^{-1/2}$$} in agreement with  \eqref{estimab}, \eqref{estab1}.

Also, $\partial_\beta b_1=h_1^{-1}(\partial_xb\frac{dx}{d\lam}+\partial_\lam b)\frac{d\lam}{d\beta}$
yields results  in agreement with  \eqref{estimab}, \eqref{estab1}. Higher order derivatives are estimated inductively.

\subsection{Proof of Lemma\,\ref{WtoAB}}\label{PfL5}

The result follows by a direct calculation based on the asymptotic of Airy functions as follows. 

We use the classical asymptotic expansion for the Airy functions as $z\to\infty$ 
$${\rm Ai}(-z)+i{\rm Bi}(-z)\sim\frac{1}{\sqrt{\pi}}z^{-1/4}\left[\cos(\zeta-\frac{\pi}{4})-i\sin(\zeta-\frac{\pi}{4})\right]=\frac{1}{\sqrt{\pi}}e^{i\pi/4}z^{-1/4}e^{-i\zeta}$$
where
$$\zeta=\frac 23 z^{3/2}=\frac 23 h_1^{-1}\eta^{3/2}=\frac{y_1^2}{2h_1}-\frac{1}{2h_1}\ln y_1-\frac{1}{4h_1}-\frac{1}{2h_1}\ln 2+O(y^{-2})$$
(since $\int_1^y \sqrt{t^2-1}\,dt=\frac 12y^2-\frac 12\ln y-\frac 14-\frac 12 \ln 2+O(y^{-2})$) and therefore
$$e^{-i\zeta}=y_1^{\frac{i}{2h_1}}e^{-i\frac{y_1^2}{2h_1}}e^{\frac{i}{4h_1}}2^{\frac{i}{2h_1}}(1+O(y_1^{-2}))$$
Furthermore,  
$$z^{-1/4}=h_1^{1/6}\eta^{-1/4}$$
Using the fact that $g\,\eta^{-1/4}=y_1^{-1/2}(1+O(y_1^{-2}))$ we obtain the result of Lemma\,\ref{WtoAB}.

\subsection{Proof of Lemma\,\ref{L5}  }\label{PFLem6} 

We assume $\hb<h_0$. 

Substituting $x=\lam^{-1}\zeta$ in equation \eqref{eqA} and denoting $\nu=h_1^{-1}$ we obtain
\bel{eqAzeta}
\frac{d^2}{d\zeta^2}\mathcal{A}=\nu^2\mathcal{A}+g^4V_1\mathcal{A}
\ee
to which we apply Lemma D.5 in \cite{CDST} (stated here, for completeness, as Lemma\,\ref{lem:bessel-} in \S\ref{PropC8}), yielding

\begin{Lemma}\label{NewlyLem}

Equation \eqref{eqAzeta} has solutions of the form 
$$[{\rm Ai}(\nu^{2/3}\zeta)\pm i{\rm Bi}(\nu^{2/3}\zeta)]\,(1+a_\pm(\zeta;\nu))$$ 
with errors having the symbol-like behavior \eqref{estimcRD}, \eqref{estimcRD1}.

\end{Lemma}

{\em Indeed,} $\zeta=-\eta(y_1)\Le 0$, and using \eqref{asyetag} it is easy to check that the assumptions of Lemma D.5 are satisfied (the proof of Lemma D.5 only uses the symbol behavior \eqref{boundsV2}, and not the particular form of $V_2$). $\Box$

\

The estimates \eqref{estimcRD}, \eqref{estimcRD1} can be straightforwardly translated into \eqref{estimc}, \eqref{estimc1}, completing the proof of Lemma\,\ref{L5}.

\subsection{Proof of Lemma\,\ref{L10}}\label{PfL9}

Solutions \eqref{solp1A},\,\eqref{solp1B} are linked to solutions \eqref{solppm} by $$\psi_\pm=\alpha_\pm\psi_{1,A}+\beta_\pm\psi_{1,B}$$ where $$\alpha_\pm=W[\psi_\pm,\psi_{1,B}]/W[\psi_{1,A},\psi_{1,B}]$$ and 
$$\beta_\pm=W[\psi_{1,A},\psi_\pm,]/W[\psi_{1,A},\psi_{1,B}].$$ Each of the four Wronskians can be easily estimated using \eqref{solp1A},\,\eqref{solp1B},\, \eqref{solppm} evaluated at $y_1=1$, yielding the conclusions of Lemma\,\ref{L10}.

\section{The case $\hb/\beta \gtrsim 1$}

In this section we assume that $|\beta|/(2\hb)\Le a_0$ for some $a_0>0$.

We change variables in \eqref{eqpsi2} as follows: with
\bel{sbbets}
y=x\sqrt{\hb/2},\ \ \psi_2(y)=\psi(x\sqrt{\hb/2})=u(x),\ \ a=\beta/(2\hb)
\ee
equation   \eqref{eqpsi2} becomes
\bel{equu}
u(x)''=\left(a-\frac{x^2}{4}\right)u(x)+\frac{\hb}{2}f\left(x\sqrt{\hb/2}\right)u(x)
\ee

\begin{Theorem}\label{betasmall}

Let $x\Ge 0$. Equation \eqref{equu} has two independent solutions of the following forms:

\bel{soluE}
u_E(x)=E\left(a,x\right)\, (1+e(x;\hb, \beta))
\ee
\bel{soluES}
u_E^*(x)=E^*\left(a,x\right)\, (1+e^*(x;\hb, \beta))
\ee
where
\bel{ese}
e,e^*=O(\hb\,\ln \hb) \ \ \text{for }x\Le \sqrt{2/{\hb}}\ \ \text{and } e,e^*=O(\hb) \  \text{for all }x\Ge \sqrt{2/{\hb}}
\ee

Also
\begin{equation}\label{soluEp}
\begin{split}
\partial_xu_E(x)&=\partial_xE\left(a,x\right)\, (1+\tilde{e}(x;\hb, \beta)) \\
\partial_xu_E^*(x)&=\partial_xE^*\left(a,x\right)\, (1+\tilde{e}^*(x;\hb, \beta))
\end{split}
\end{equation}
where $\tilde{e}, \tilde{e}^* $ satisfy:
\bel{eq72.5}
{ |\tilde{e}|\lesssim \hb \la x\ra^{-1}}
\ee

Furthermore, $e=O(x^{-2})=O(\hb)$ for ${x>\sqrt{2/\hb}}$, and the derivatives satisfy the estimates: for $k,\ell\Ge 0$ we have 
\bel{likeasym}
\begin{array}{ll}
 |\partial_x^{k+1}\partial_\beta^\ell e(x;\hb, \beta)|\lesssim    x^{-3-k}\hb^{-\ell}<{x^{-1-k} \hb^{-\ell+1}}&  \text{for }{x>\sqrt{2/\hb}}\\
  \\
   |\partial_x^{k+1}\partial_\beta^\ell e(x;\hb, \beta)|\lesssim x^{-1-k} {\hb^{-\ell+1}}& \text{for } x\in[\sqrt{2}, {\sqrt{2/\hb}}] \\
   \\ 
 |\partial_x^{k+1}\partial_\beta^\ell e(x;\hb, \beta)|\lesssim   {\, \hb^{-\ell+1}} & \text{for }x\in[0, \sqrt{2}]
 \end{array}
\ee

\end{Theorem}

\ 

The {\em Proof of Theorem\,\ref{betasmall}} is presented in \S\ref{esePf}-\S\ref{Pf73}. Section \S\ref{Monbsma} contains the monodromy and the scattering matrix in this case.

\subsection{Proof of \eqref{ese}}\label{esePf}
Denote for short $e(x;\hb, \beta)=e(x)$.

Substituting \eqref{soluE} in \eqref{equu}, we obtain the integral equation:  
\begin{multline}
\label{eq2}
e(x)=\int_{\infty}^x \frac1{E(a,s)^2}\int_{\infty}^s\frac{\hbar}2\, f(t\sqrt{{\hbar}/2})\, E(a,t)^2\, (1+e(t))\, dt\,ds
\\
=\int_{\infty}^x (1+e(t))\, \frac{\hbar}2\,  f(t\sqrt{{\hbar}/2})\,  E(a,t)^2\, \Big(\int_x^t\frac{ds}{E(a,s)^2}\Big)\, dt
\end{multline}
and using
$$\Big(\frac{E^*}{E}\Big)'=\frac{W[E,E^*]}{E^2}=\frac{-2i}{E^2},$$
the equation becomes:
\begin{multline}
\label{eq3}
={\frac {i\hbar}4}\int_{\infty}^x (1+e(t))\,f(t\sqrt{{\hbar}/2})\, E(a,t)^2\Big(\frac{E^*(a,t)}{E(a,t)}-\frac{E^*(a,x)}{E(a,x)}\Big)dt
\\
={\frac {i\hbar}4}\int_{\infty}^x (1+e(t))\,f(t\sqrt{{\hbar}/2})\, \Big(|E(a,t)|^2-E(a,t)^2\frac{E^*(a,x)}{E(a,x)}\Big)dt
\\
=:F(x)+[Ge](x)=:J[e](x)
\end{multline}
where
$$F(x)=\int_{\infty}^x K(x,t)\, dt,\ \ \ G[e](x)=\int_{\infty}^x K(x,t)\, e(t)\, dt$$
with
$$K(x,t)=\frac{i\hb}{4} f(t\sqrt{{\hbar}/2})\Big(|E(a,t)|^2-E(a,t)^2\frac{E^*(a,x)}{E(a,x)}\Big)$$

We use the following estimates (we use \eqref{estimf}, \eqref{asyE}, \eqref{errsymE}): for all $x\Ge 0$ and $a$ with $|a|\Le a_0$} (for $\hb<\text{Const.}$) we have: \newline
- for $t\in[0,1]$ we have $|E(a,t)|<C_1$ and $|f(t\sqrt{{\hbar}/2})|<C_2$ therefore
$$| K(x,t)|\Le\hb\,  C_2C_1^2/2\ \ \ \  \text{for all }t\in[0,1],\ x\Ge 0$$
-for $t\in[1,\sqrt{2/\hb}]$ we have $|E(a,t)|<C_3(=M_0)t^{-1/2}$ and $|f(t\sqrt{{\hbar}/2})|<C_2$ therefore
$$| K(x,t)|\Le\hb\, \frac{1}{t}\, C_2C_3^2/2\ \ \ \  \text{for all }t\in [1,\sqrt{2/\hb}],\ x\Ge 0$$
-for $t\Ge \sqrt{2/\hb}$ we have $|E(a,t)|<C_3t^{-1/2}$ and, from \eqref{estimf}, $|f(t\sqrt{{\hbar}/2})|<{C_4}/({\hb t^2})$ therefore
$$| K(x,t)|\Le \frac{1}{t^3}\, C_4C_3^2/2\ \ \ \  \text{for all }t\Ge\sqrt{2/\hb},\ x\Ge 0$$

Then for $x\Ge 0$,
\begin{multline}\label{fifty3}
|F(x)|\le\Big(\int_0^1+\int_1^{\sqrt{2/{\hbar}}}+\int_{\sqrt{2/{\hbar}}}^{\infty}\Big) \Big|
\hbar f(t\sqrt{{\hbar}/2})\left(|E(a,t)|^2-E(a,t)^2\frac{E^*(a,x)}{E(a,x)}\right)\Big| \, dt
\\ 
\le \text{Const. }\hbar +\text{Const. }\hbar\int_1^{\sqrt{2/{\hbar}}}\frac{1}t\,dt+\text{Const. }\int_{\sqrt{2/{\hbar}}}^{\infty}\, \frac{1}{t^3}\,dt
\\
=\tilde{C_1}\hbar+\tilde{C_2}\hbar\ln(\hbar^{-1})
\end{multline}
where $\tilde{C_1}$ and $\tilde{C_2}$ are independent of $\hbar$ and $a$.
\\

{\bf Remark} {\em in support of the order of the error $\hb\ln\hb$ in the second integral {of \eqref{fifty3}}.}

The estimate in the second integral seems optimal: 

(i) Denoting $F=|E(a,x)|$, $E(a,x)=Fe^{i\chi}$ we have
$$|E(a,t)|^2-E(a,t)^2\frac{E^*(a,x)}{E(a,x)}=F^2\left(1-e^{2i[\chi(t)-\chi(x)]}\right)$$
so there are no cancellations due to oscillations.

(ii) Estimating $|f(y)|\lesssim\frac1{1+y^2}$ and ${F^2}\lesssim\frac{1}{1+x}$ the second integral is estimated by
$$\hb \int_1^{\sqrt{2/{\hbar}}} { \frac1{1+\hb t^2/2}}\, \frac{1}{1+t}\,dt=O(\hb\ln\hb)$$
after an explicit calculation.

\medskip

Using the norm $\|e\|:=\sup_{x\Ge 0}|e(x)|$, we have, using similar estimates for $G[e]$,
\begin{multline}
\|J(e)\|\le \|F\|+\|G(e)\|
\le \tilde{C_1}\hbar+\tilde{C_2}\hbar\ln(\hbar^{-1})+(\tilde{C_3}\hbar+\tilde{C_4}\hbar\ln(\hbar^{-1}))\|e\|.
\end{multline}
Let $\|e\|\le \hbar\ln(\hbar^{-1})R$, $R=2(\tilde{C_1}+\tilde{C_2})$ for $\hbar$ small enough so that $\ln(\hbar^{-1})>1$ and $\tilde{C_3}\hbar+\tilde{C_4}\hbar\ln(\hbar^{-1})\le \frac12$. (The choice of $R$ is made to be independent of $\hbar$.)

Consider the closed ball $$B:=\{e(x):\|e\|\le \hbar\ln(\hbar^{-1})R\}$$ in the Banach space of continuous, bounded functions on $[0,+\infty)$. 
The mapping $e\mapsto F+Ge$ is a contractive mapping from $B$ to itself since if $e\in B$ then
$$\|J(e)\|\le \tilde{C_1}\hbar+\tilde{C_2}\hbar\ln(\hbar^{-1})+(\tilde{C_3}\hbar+\tilde{C_4}\hbar\ln(\hbar^{-1})) \hbar\ln(\hbar^{-1})R \le \hbar\ln(\hbar^{-1})R$$
and
\begin{equation}
\|J(e_1)-J(e_2)\|=\|G(e_1)-G(e_2)\|\le (\tilde{C_3}\hbar+\tilde{C_4}\hbar\ln(\hbar^{-1}))\|e_1-e_2\|\le \frac12\|e_1-e_2\|.
\end{equation}
The calculation shows that the error term for the bigger $x\ge \sqrt{2/{\hbar}}$ is $O(\hbar)$ instead of $O(\hbar\ln(\hbar^{-1}))$. Note also that $e(x)=O(x^{-2})$ for $x\ge \sqrt{2/{\hbar}}$.

This completes the proof that equation \eqref{eq3}, $e=F+Ge$, has a continuous solution: $e=(I-G)^{-1}F$ with $\|e\|<\hbar\ln(\hbar^{-1})R$ for some constant $R$. Since both $u_E$ and $E(a,x)$ are of class $C^v$ in $x$, then so is $e$. Finally, since both the function $F$ and the operator $G$ depend analytically on the parameter $a$ for $|a|<a_0$ then $e$ is also analytic in $a$ (regularity in $a$ is re-obtained below.)

\subsection{Proof of \eqref{soluEp}.}

In \eqref{equu} we substitute $u(x)=\exp(\int G)$ (note that $G=u'/u$) and obtain the equation
\bel{eqG}
G'+G^2=a-\frac{x^2}{4}+\frac 12\hb f(x\sqrt{\hb/2})
\ee

Let $G_0(x)=E'(a,x)/E(a,x)$ whence $G_0'{+G_0^2}=a-\frac{x^2}{4}$. Denoting $G=G_0+\phi$, then $\phi$ satisfies
\bel{difeqphi}
\phi'+2G_0\phi=\frac 12\hb f(x\sqrt{\hb/2})-\phi^2
\ee
or, in integral form
\bel{eqphi}
\phi(x)=E(a,x)^{-2}\int_{+\infty}^x\left[\frac 12\hb f(x\sqrt{\hb/2})E(a,t)^2-\phi^2(t)E(a,t)^2\right]\, dt
\ee

\begin{Lemma}\label{Olema}
Equation \eqref{eqphi} has a unique solution satisfying:

(i) $|\phi(x)|< 2C x^{-3}<C\hb x^{-1}$ for $x>\sqrt{2/h}$

(ii) $|\phi(x)|<C\hb x^{-1}$ for $x\in[\sqrt{2}, \sqrt{2/h}]$

(iii) $|\phi(x)|<C\hb $ for $x\in[{0}, \sqrt{2}]$

{(iv) $\phi(x)$ behaves like a symbol.}

{(v) $\partial_\beta^\ell\phi(x)$ satisfies the estimates at (i)...(iii) multiplied by $\hb^{-\ell}$.}

\end{Lemma}

The proof is found in \S\ref{PfOlema} below. Let us first show that this implies \eqref{soluEp}. 

The relation $G=G_0+\phi$ is, in fact $u'/u=E'/E+\phi$. Hence, we have found a solution $u(x)$ so that
$$\ln u(x)=\ln E(a,x)+\int_\infty^x \phi(t)\, dt.$$ Therefore, $u(x)=E(a,x) (1+e(x))$ where 
\bel{eeephi}
1+e(x)=\exp(\int_\infty^x \phi(t)\, dt)
\ee 
showing that $e(x)=O(x^{-2})={O(\hb)} $ for $x>\sqrt{2/\hb}$. In conclusion,  $u(x)$ is the solution we found in \S\ref{esePf}.

Note that from \eqref{eeephi} and \eqref{eqphi} it is easy to see that $e$ is of class $C^v$.

On the other hand, we have
\begin{multline}
u'(x)=u(x)\left( \frac{E'(a,x)}{E(a,x)} +\phi(x)\right)=E(a,x)(1+e(x))\left( \frac{E'(a,x)}{E(a,x)} +\phi(x)\right)\\  =E'(a,x)\left(1+\frac{E(a,x)}{E'(a,x)}\phi(x)+e(x)+\frac{E(a,x)}{E'(a,x)}\phi(x)e(x)\right)
\end{multline}
which implies \eqref{soluEp}.

\subsubsection{ Proof of Lemma\,\ref{Olema}.}\label{PfOlema}

Denoting $\xi=x^2/2$, $\phi(x)=\phi(\sqrt{2\xi}):=\tilde{\phi}(\xi)$ and changing the integration variable to $\tau=t^2/2$ equation \eqref{eqphi} becomes
\begin{multline}\label{eqphit1}
\tilde{\phi}(\xi)=E(a,\sqrt{2\xi})^{-2}\int_{+\infty}^\xi\left[\frac 12\hb f(\sqrt{\hb\tau})E(a,\sqrt{2\tau})^2-\tilde{\phi}^2(\tau) {E(a,\sqrt{2\tau})^2}\right]\, \frac{d\tau}{\sqrt{2\tau}}\\
=E(a,\sqrt{2\xi})^{-2}\int_{+\infty}^\xi  \frac 12\hb f(\sqrt{\hb\tau})E(a,\sqrt{2\tau})^2\, \frac{d\tau}{\sqrt{2\tau}}-E(a,\sqrt{2\xi})^{-2}\int_{+\infty}^\xi\tilde{\phi}^2(\tau) {E(a,\sqrt{2\tau})^2}\, \frac{d\tau}{\sqrt{2\tau}}\\
:=F_0(\xi)+\tilde{J}_0\tilde{\phi}(\xi):=\tilde{J}\tilde{\phi}(\xi)
\end{multline}

\

{\em Proof of (i).}

\

We show that operator $\tilde{J}$ is contractive in the Banach space $\mathcal{B}_1$ of continuous functions $\tilde{\phi}$ on the interval $[\hb^{-1},\infty)$ equipped with the norm 
\bel{norm1}
\|\tilde{\phi}\|_1=\sup_{\xi\Ge\hb^{-1}}  \xi^{3/2}|\tilde{\phi}(\xi)|
\ee

Indeed, if $\tilde{\phi}\in\mathcal{B}_1$ then using \eqref{symbE} we have
$$2\sqrt{2}F_0(\xi)=\xi^{\frac 12+ia}e^{-i\xi}(1+\tilde{e}_E)^{-2}\int_{+\infty}^\xi   \hb  f(\sqrt{\hb\tau}) \tau^{-1-ia}   e^{i\tau}(1+\tilde{e}_E)^{2}
{d\tau}$$
and integrating by parts this equals 
\begin{multline}\label{eq88}
\xi^{-\frac 12+ia}  \hb f(\sqrt{\hb\xi}) \\
- \xi^{\frac 12+ia}e^{-i\xi}(1+\tilde{e}_E)^{-2}\int_{+\infty}^\xi  {d\tau}  \hb  \left[ \frac{\sqrt{\hb}}{2\sqrt{\tau}}f'(\sqrt{\hb\tau}) (1+\tilde{e}_E)^{2} + f(\sqrt{\hb\tau}) (1+\tilde{e}_E)\tilde{e}_E'\right]\int_\infty^\tau dt\, t^{-1-ia}   e^{it}
\end{multline}
where from \eqref{estimf} and the fact that $|\int_\infty^\tau dt\, t^{-1-ia}   e^{it}|\lesssim \tau^{-1}$  we infer that $|F_0|<C_1\xi^{-3/2}$.


Integrating by parts, \eqref{eq88} becomes
$$=-i\hb f(\sqrt{\hb\xi})\xi^{-\frac12}
+i\hb\xi^{\frac 12+ia}e^{-i\xi}(1+\tilde{e}_E)^{-2}
\int_{+\infty}^\xi  {d\tau} e^{i\tau}\left[f(\sqrt{\hb\tau}) \tau^{-1-ia}   e^{i\tau}(1+\tilde{e}_E)^{2}\right]'. $$
Since $f(y)$ and $\tilde e_E(\tau)=e_E(t)$ behave like a symbol and in view of \eqref{estimf}, we obtain $|F_0|<C_1\xi^{-3/2}$.

Using \eqref{symbE} and \eqref{norm1} we obtain
$$|\tilde{J}_0\tilde{\phi}(\xi)|\lesssim \xi^{\frac 12}{|1+\tilde{e}_E|^{-2}}\int_\xi^{+\infty}   \tau^{-3} \tau^{-1}   {|1+\tilde{e}_E|^{2}}
{d\tau}\ \|\tilde{\phi}\|_1^2\Le C_2\xi^{-5/2}\|\tilde{\phi}\|_1^2<C_2\hb\xi^{-3/2}\|\tilde{\phi}\|_1^2$$

Consider the ball $\|\tilde{\phi}\|_1\Le R$.
A similar estimate shows that $\tilde{J}$ is a contraction if $2C_2\hb R<1$. 
The ball  is invariant under $\tilde{J}$ if $C_1+C_2\hb R^2\Le R$ and both conditions are clearly possible  (for $\hb$ small enough).

In particular, the solution satisfies 
$$\sup_{\xi\Ge \hb^{-1}} |\tilde{\phi}| \Le {\rm Const.}\ \xi^{-3/2}<{\rm Const.}\ \hb\xi^{-1/2}$$
{which is equivalent to (i).}

\

{\em Proof of (ii).}

\

For $\xi\in [1,\hb^{-1}]$ we rewrite the equation {\eqref{eqphit1}} as 
\begin{multline}\label{eqphit2}
\tilde{\phi}(\xi)=\frac{E(a,\sqrt{2\hb^{-1}})^2}{E(a,\sqrt{2\xi})^2}\,\tilde{\phi}(\hb^{-1})+
E(a,\sqrt{2\xi})^{-2}\int_{\hb^{-1}}^\xi\left[\frac 12\hb f(\sqrt{\hb\tau})E(a,\sqrt{2\tau})^2-\tilde{\phi}^2(\tau)E(a,\tau)^2\right]\, \frac{d\tau}{\sqrt{2\tau}}\\
=\frac{E(a,\sqrt{2\hb^{-1}})^2}{E(a,\sqrt{2\xi})^2}\,\tilde{\phi}(\hb^{-1})+
E(a,\sqrt{2\xi})^{-2}\int_{\hb^{-1}}^\xi  \frac 12\hb f(\sqrt{\hb\tau})E(a,\sqrt{2\tau})^2\, \frac{d\tau}{\sqrt{2\tau}}\\
-E(a,\sqrt{2\xi})^{-2}\int_{\hb^{-1}}^\xi\tilde{\phi}^2(\tau)E(a,\tau)^2\, \frac{d\tau}{\sqrt{2\tau}}\\
:=F_1(\xi)+\tilde{J}_1\tilde{\phi}(\xi):=\tilde{J}_2\tilde{\phi}(\xi)
\end{multline}
and we show that $\tilde{J}_2$ is contractive in the Banach space  $\mathcal{B}_2$ of continuous functions $\tilde{\phi}$ on the interval $[1,\hb^{-1}]$ equipped with the norm 
\bel{norm2}
\|\tilde{\phi}\|_2=\sup_{\xi\in [1,\hb^{-1}]}  \xi^{1/2}|\tilde{\phi}(\xi)|
\ee

The estimates are similar to those at point (i), except that on this interval we use the fact that $f$ is bounded, and we obtain that the solution is $O(\hb\xi^{-1/2})$.

\

{\em Proof of (iii).}

\

For $\xi\in [0,1]$ we rewrite the equation {\eqref{eqphit1}} as 
\begin{multline}\label{eqphit3}
\tilde{\phi}(\xi)=\frac{E(a,\sqrt{2})^2}{E(a,\sqrt{2\xi})^2}\,\tilde{\phi}(1)+
E(a,\sqrt{2\xi})^{-2}\int_{1}^\xi\left[\frac 12\hb f(\sqrt{\hb\tau})E(a,\sqrt{2\tau})^2-\tilde{\phi}^2(\tau)E(a,\tau)^2\right]\, \frac{d\tau}{\sqrt{2\tau}}\\
=\frac{E(a,\sqrt{2})^2}{E(a,\sqrt{2\xi})^2}\,\tilde{\phi}(1)+
E(a,\sqrt{2\xi})^{-2}\int_{1}^\xi  \frac 12\hb f(\sqrt{\hb\tau})E(a,\sqrt{2\tau})^2\, \frac{d\tau}{\sqrt{2\tau}}\\
-E(a,\sqrt{2\xi})^{-2}\int_{1}^\xi\tilde{\phi}^2(\tau)E(a,\tau)^2\, \frac{d\tau}{\sqrt{2\tau}}\\
:={F_2(\xi)+\tilde{J}_3\tilde{\phi}(\xi):=\tilde{J}_4\tilde{\phi}(\xi).}
\end{multline}


 We show that $\tilde{J}_4$ is contractive in the Banach space  $\mathcal{B}_3$ of continuous functions $\tilde{\phi}$ on the interval $[0,1]$ equipped with the sup norm 
\bel{supnorm}
\|\tilde{\phi}\|_{\infty}=\sup_{\xi\in [0,1]} |\tilde{\phi}(\xi)|
\ee
The estimates are similar to those in parts (i) and (ii), except that on this interval we use the fact that not only $f$ is bounded, but also $C_1\le|E(a,\sqrt{2\xi})|\le C_2$ for some positive constants $C_1$ and $C_2$, and we obtain that the solution is $O(\hb)$.

\

{ \em Proof of (iv).}
This results directly using the differential equation \eqref{difeqphi} and the fact that $G_0$ behaves like a symbol, see \eqref{symbE}.

\

 {\em Proof of (v).}
 Differentiating equation \eqref{difeqphi} with respect to $\beta$ we obtain
$$\partial_\beta \phi'+(2G_0+2\phi)\partial_\beta \phi=\frac 12\hb \partial_\beta f(x\sqrt{\hb/2})-\frac{1}{\hb}\partial _a G_0\, \phi$$
and using \eqref{eeephi} 

$$\partial_\beta\phi(x)=E(a,x)^{-2}(1+e(x))^{-2}\int_{+\infty}^x\left[\frac 12\hb \partial_\beta {f(t\sqrt{\hb/2})}-\frac{1}{\hb}\partial _a G_0\, \phi\right]E(a,t)^2(1+e(t))^{2}\, dt$$
where the integral has the same form as the nonhomogeneous term in \eqref{eqphi}, only the integrand has a factor ${\hb}^{-1}$.

Since $\partial_\beta^{\ell}G_0(x)\sim x$ for $x\to \infty$ (see Lemma\,\ref{DaG0}) then using (i)...(iii) we see that $G_0\, \phi$ has the same behavior as $\hb f(x\sqrt{\hb/2})$ hence the same estimates as for the nonhomogeneous term in the proof of (i)...(iii) apply, yielding the same result only multiplied by ${\hb}^{-1}$.

Higher order derivatives are estimated by induction on $\ell$. Differentiating equation \eqref{difeqphi} $\ell$ times with respect to $\beta$ we obtain
\begin{multline}
\partial_\beta^\ell \phi'+(2G_0+2\phi)\partial_\beta^\ell \phi=\frac 12\hb \partial_\beta^\ell f(x\sqrt{\hb/2})\\
-2\sum_{j=1}^{\ell-1} \binom{\ell}{j}    (2\hb)^{-j}\partial _a^j G_0\, \partial_\beta^{\ell-j}\phi-\sum_{j=1}^{\ell-1} \binom{\ell}{j}   \partial_\beta^{j}\ \, \partial_\beta^{\ell-j}\phi
\end{multline}

As for the first derivative also for arbitrary $\ell$ the same estimates as for the nonhomogeneous term in the proof of (i)...(iii) apply, yielding the same result only (after using the induction hypothesis) multiplied by ${\hb}^{-\ell}$.

\subsection{Proof of \eqref{likeasym}}\label{Pf73}

Differentiating in \eqref{eeephi} we obtain $\partial_xe=\phi(1+e)$ therefore $\partial_xe$ satisfy the same estimates as $\phi$ given in Lemma\,\ref{Olema}. The higher derivatives are estimated by a straightforward induction.

Also from \eqref{eeephi} we have  $$\partial_\beta e=\left(\int_\infty^x \partial_\beta\phi(t)\, dt\right)\ (1+e)$$  and using Lemma\,\ref{Olema} (v) followed by (i)-(iii) we obtain the stated estimate for $\partial_\beta e$.

Estimates for higher order derivatives are found similarly, by induction.

This completes the proof of Theorem\,\ref{betasmall}.

\subsection{Solutions for $x\Le0$} Arguing as in \S\ref{y1neg} we obtain as a consequence of Theorem\,\ref{betasmall} that
\begin{Corollary}\label{C10}
Let $x\Le 0$. Equation \eqref{equu} has two independent solutions $u_E^\ell(x)$ and $u_E^{*\,\ell}(x)=\overline{u_E^\ell(x)}$ satisfying
$$u_E^\ell(x)=E\left(a,-x\right)\, (1+e^\ell(x;\hb, \beta))$$
where $e^\ell(x;\hb, \beta)$ satisfies \eqref{ese}, \eqref{eq72.5} for $x\to -\infty$ and \eqref{likeasym} with $x$ replaced by $-x(=|x|)$.
\end{Corollary}

\subsection{Matching at $x=0$ and the scattering matrix}\label{Monbsma} 

The matching, monodromy and scattering matrix is deduced as in \S\ref{Match0}. As expected, the dominant term of the monodromy matrix turns out to be exactly \eqref{monof}. The rest of this section shows the main steps of the calculation which leads to this result.

Working back through the substitutions \eqref{sbbets} the solution $u_E(x)$ in \eqref{soluE} corresponds to
$$\psi(\xi(y))=\sqrt{\xi'}\psi_2(y)=\sqrt{\xi'}u_E(y\sqrt{2/\hb})$$
We have
$$u_E(y\sqrt{2/\hb})=E(\frac{\beta}{2\hb},y\sqrt{2/\hb})(1+O(y^{-2}))=2^{1/4}\hb^{\frac 14+i\frac{\beta}{4\hb}} e^{i\pi/4+i\phi_2/2}y^{-\frac 12-i\frac{\beta}{2\hb}}e^{i\frac{y^2}{2\hb}}(1+O(y^{-2}))$$
and therefore, as in \S\ref{ScatMat},
$$\sqrt{\xi'}u_E(y\sqrt{2/\hb})=\tilde{B} e^{i\tilde{\phi}}y^{-i\frac{\beta}{2\hb}}e^{i\frac{y^2}{2\hb}}(1+O(y^{-2})),\ \text{where }\tilde{B} \in\RR,\ e^{i\tilde{\phi}}=\hb^{i\frac{\beta}{4\hb}} e^{i\pi/4+i\phi_2/2} $$

Comparing to \eqref{fpmr}, \eqref{fpml} we see that
$$ f_+^r(\xi(y))=\tilde{K}_+  \sqrt{\xi'}\,  u_E(y\sqrt{2/\hb}),\ \text{where }\tilde{K}_\pm=\frac{1}{\tilde{B} e^{\pm i\tilde{\phi}}}\, e^{\pm i\frac{\sqrt{E}C_+}{\hb}}$$
and
$$ f_+^\ell(\xi(y))=\tilde{K}_-^\ell \sqrt{\xi'}u_E^*(y\sqrt{2/\hb}),\ \text{where }\tilde{K}_\pm^\ell=\frac{1}{\tilde{B} e^{\pm i\tilde{\phi}}}e^{\pm i\frac{\sqrt{E}C_-}{\hb}}$$

Of course, $f_\pm^{r,\ell}(\xi(y))$ is the complex conjugate of $f_\mp^{r,\ell}(\xi(y))$.

As in \S\ref{ScatMat} a direct calculation gives

$$ \frac 1{\sqrt{\xi'(y)}}F^{\ell}:=\frac 1{\sqrt{\xi'(y)}}[f_+^{\ell},f_-^{\ell}]=  \frac 1{\sqrt{\xi'(y)}}F^r \,(I+\hb\ln\hb^{-1}\mathcal{R}_1)\mathcal{M}_0(I+\hb\ln\hb^{-1}\mathcal{T}_1) 
$$

$$ \frac 1{\sqrt{\xi'(y)}}F^r:=\frac 1{\sqrt{\xi'(y)}}[f_+^r,f_-^l]=  \frac 1{\sqrt{\xi'(y)}}F^\ell \,(I+\hb\ln\hb^{-1}\mathcal{R}_1)\mathcal{M}_0(I+\hb\ln\hb^{-1}\mathcal{T}_1) 
$$
with $\mathcal{M}_0$ given by \eqref{monof} and the matrices $\mathcal{R}_1,\, \mathcal{T}_1$ have entries which are rational functions in $e(0,\hb,\beta)$ and $e^*(0,\hb,\beta)$ and therefore satisfy, by \eqref{likeasym}, $|\partial_\beta^\ell R_{ij}|\lesssim   \hb^{-\ell+1}$.

It follows that the dominant behavior of the scattering matrix is the same as in \S\ref{ScatMat}, and we obtain
 \eqref{S11}, \eqref{S21} for $\beta\Ge 0$ and, for $\beta<0$ we arrive, as in  \S\ref{ScatMat3}, at the formulas  \eqref{S11n}, \eqref{S21n}.

\section{The case $\beta<0$ with $\hb/|\beta|<Const.$}\label{Bcase}

Denote $-\beta=B>0$. Equation \eqref{eqpsi2} becomes
\bel{eqnegbeta}
\frac{d^2{\psi_2}}{dy^2}=-\hb^{-2}(B+y^2){\psi_2}+f(y){\psi_2}
\ee

In this case there are no turning points, and the behavior of the Weber functions (when $f\equiv 0$) is purely oscillatory. We will approximate solutions using the Airy functions, similar to the approach in \S\ref{Aicase}. 

We substitute
\bel{psi123}
h_3=\hb/B,\ \ y=\sqrt{B}y_3,\ \ {\psi_2}(y):=\psi_3(y_3),\ \ B f(\sqrt{B}y_3):=f_3(y_3)
\ee
which transforms \eqref{eqnegbeta} to
\bel{eqpsi3}
\frac{d^2\psi_3}{dy_3^2}=-h_3^{-2}(1+y_3^2)\psi_3+f_3\psi_3
\ee
which can be viewed as a perturbation of the Weber equation
\bel{eqw3}
\frac{d^2w}{dy_3^2}=-h_3^{-2}(1+y_3^2)w
\ee

Denote, for any $\delta>0$, 
\bel{eta3}
\eta_3(y_3)=\left(\frac 32\int_{-\delta}^{y_3}\sqrt{t^2+1}\, dt\right)^{2/3}\ {\rm for \ }y_3\Ge 0\ee
and let
\bel{defg3}
g_3(y_3)=\left(\frac{\eta_3(y_3)}{y_3^2+1}\right)^{1/4}
\ee

These functions satisfy relations similar to \eqref{useide}, more precisely
\bel{useide3}
g_3^2\eta_3'=1,\ \ \eta_3\eta_3'^2=y_3^2+1,\ \ \eta_3/g_3^4=y_3^2+1
\ee
and note that $\eta_3(y_3)$ is $C^\infty([0,+\infty))$, is increasing, and
\bel{asyetag3}
\eta_3(y_3)\sim (3/4)^{2/3}y_3^{4/3}, \ \ { g_3^4(y_3)\sim (3/4)^{2/3}\, y_3^{-2/3}  \ \ \ \text{for }y_3\to\infty}
\ee

We proceed as in \S\ref{PFLem6}: 
substituting
$$\psi_3(y_3)=g_3(y_3)F( -h_3^{-2/3}\eta_3(y_3)),\ \ -\eta_3(y_3)=\zeta,\ \ \nu=h_3^{-1}$$
equation \eqref{eqpsi3} becomes
$$\frac{d^2}{d\zeta^2}F''=\nu^2\zeta F+g_3^4V_3F,\ \ \ \text{where }V_3=f_3-\frac{g''}g$$
to which we apply Lemma D.5 in \cite{CDST} (Lemma\,\ref{lem:bessel-} in \S\ref{PropC8}). In the present case $\zeta\Le-\eta(0)<0$ but the proof of Lemma D.5 in \cite{CDST} goes through as such. Working back through the substitutions we obtain the Jost solutions of \eqref{eqpsi3}:

\bl\label{L53} Denote
\bel{defAB3}
A_3(y_3)=g_3(y_3){\rm Ai}\left(-h_3^{-2/3}\eta_3(y_3)\right),\ \ B_3(y_3)=g_3(y_3){\rm Bi}\left(-h_3^{-2/3}\eta_3(y_3)\right)
\ee
where ${\rm Ai},\,{\rm Bi}$ are the Airy functions.

(i) For $y_3\Ge 0$ eq. \eqref{eqpsi3} has two independent solutions of the form
\bel{solppm3}
\psi_{3,\pm}(y_3)=\left[A_3(y_3)\mp i B_3(y_3)\right]\,  \left(1+h_3 c_\pm(y_3;h_3,B)\right)
\ee
where
\bel{estimc3}
\big|\partial_{y_3}^k\partial_{B}^l c_\pm\big|\Le C_{kl}\, \la y_3\ra^{-2-k}B^{-l}
\ee
Also, at $y_3=0$:
\bel{estimc13}
\big|\partial_{\beta}^l c_\pm(0,h_3,\beta)\big|\Le \,C_{l}\,\beta^{-l},\ \ \big|\partial_{\beta}^l \partial_{y_3}\, c_\pm(0,h_3,\beta)\big|\Le h_3^{-2/3}\,C_{l}\,B^{-l}
\ee

(ii) Furthermore:
$$\psi'_{3,\pm}(y_3)=\left[A'_3(y_3)\mp i B'_3(y_3)\right]\,  \left(1+h_3 {c^d_\pm(y_3;h_3,B)}\right)$$
where the error terms ${c^d_\pm}$ satisfy estimates similar to \eqref{estimc3},\,\eqref{estimc13} for $y_3\Ge 0$.
\el

In particular, for $f\equiv 0$, we obtain:
\begin{Corollary}\label{C63}
 The Weber equation \eqref{eqw3} has two solutions $w_{3,\pm}$ estimated as in Lemma~\ref{L53}.
\end{Corollary}
From Lemma\,\eqref{L53} and {Corollary}\,\ref{C63} it follows that
\begin{Corollary}\label{C73}
$$\psi_{3,\pm}=w_{3,\pm}\left(1+h_3 \tilde{c}_\pm(y_3;h_3,\beta)\right)$$
with $\tilde{c}_\pm$ satisfying \eqref{estimc3},\eqref{estimc13}.
\end{Corollary}

Analogous to Lemma\,\ref{WtoAB}:
\bl\label{WtoAB3} We have 
$$A_3(y_3)\pm iB_3(y_3)=g(y_3)({\rm Ai}\pm i{\rm Bi})(-h_3^{-2/3}\eta_3(y_3))\sim \, \lambda_{3,\pm} \ y_3^{-\frac 12\mp\frac{ i}{2h_3}}\,e^{\mp iy_3^2/(2h_3)}\ \ \ (y_3\to +\infty)$$
where
$$\lambda_{3,+}=\overline{\lambda_{3,-}}=\pi^{-1/2}h_3^{1/6}e^{i\pi/4}\, (4e)^{-i/4h_3}\, e^{-iC_3/h_3}$$
with
$$C_3=\frac 12\,\delta\,\sqrt {{\delta}^{2}+1}+\frac 12\ln\left(\delta+\sqrt {{\delta}^{2}+1}\right)$$
\el

\begin{proof}The proof of Lemma\,\ref{WtoAB3} is almost the same as that of Lemma\,\ref{WtoAB}, contained in \S\ref{PfL5}. \end{proof}

Similar to Lemma\,\ref{L5}, and we the same proof, we have

\subsection{The regions with  $y_3\Le 0$}\label{y1neg3}
The same argument as in \S\ref{y1neg} gives:

\bl 

(i) For $y_3\Le 0$ eq. \eqref{eqnegbeta}  has two independent solutions of the form
$$\psi^\ell_{3,\pm}(y_3)=\left[A_3(-y_3)\mp i B_3(-y_3)\right]\,  \left(1+h_3 c^\ell_\pm(y_3;h_3,B)\right)$$
where $c^\ell_\pm$ satisfy \eqref{estimc3},\eqref{estimc13}.

(ii) In particular, the Weber equation \eqref{eqw3} also has two solutions $w^\ell_{3,\pm}$ of the form in (i), therefore
$$\psi^\ell_{3,\pm}=w^\ell_\pm\left(1+h_3 \tilde{c}^\ell_\pm(y_3;h_3,B)\right)$$
with $\tilde{c}^\ell_\pm$ satisfying \eqref{estimc3},\eqref{estimc13}.
\el

\

\subsection{Matching at $y_3=0$} As in \S\ref{Match0} we have
\bp\label{Prop133}
We have 
\bel{MonP3}
[\psi_{3,+}^\ell\ \psi_{3,-}^\ell]=[\psi_{3,+}\ \psi_{3,-}]\  N_3,\ \ {\rm with\ } N_3=(I+h_3R)M_3(I+h_3T)
\ee 
where $M_3$ is the monodromy matrix of the Weber equation \eqref{eqw3}, given by \eqref{monoWM3}, and the matrices $R,T$ have bounded entries in the parameters for $h_3{\lesssim 1}, 
\beta<\beta_0$.
\ep

The proof is straightforward, using arguments similar to those in \S\ref{Match0}.

\subsection{The scattering matrix}\label{ScatMat3}

The Jost solutions of \eqref{eqf} satisfy \eqref{fpmr}, \eqref{fpml}.

{On the other hand, we work} back through the substitutions $f(\xi(y))=\sqrt{\xi'(y)}{\psi_2}(y)$, followed by \eqref{psi123}. Using \eqref{asydxi} and Corollary\,\ref{C73} we see that
$$ f_{\pm}^r(\xi(y))=K_\pm\, \sqrt{\xi'(y)}\,  \psi_{3,\pm}(y\sqrt{|\beta|})=K_{3,\pm}\, \sqrt{\xi'(y)}\,  w_{3,\pm}(y\sqrt{|\beta|})(1+O(y^{-2})),$$
where $$K_{3,\pm}= \frac{ E^{1/4} e^{\pm i\frac{\sqrt{E}C_{+}}{\hb}}}{\lambda_{3,\mp}\, |\beta|^{\frac 14\pm\frac{i\beta}{4\hb}}}.$$
Similarly, 
$$ f_{\pm}^\ell(\xi(y))=K_{3,\mp}^\ell\, \sqrt{\xi'(y)}\,  \psi_{3,\mp}^\ell(y\sqrt{|\beta|}) ,\ \text{with }K_{3,\pm}^\ell= \frac{ E^{1/4} e^{\pm i\frac{\sqrt{E}C_{-}}{\hb}}}{\lambda_{3,\mp}\, |\beta|^{\frac 14\pm\frac{i\beta}{4\hb}}}$$
We now use Proposition\,\ref{Prop133}, and calculating as in \S\ref{ScatMat}  it follows that 

$${[f_+^{\ell}\, f_-^{\ell}]}=[f_+^{r}\, f_-^{r}]\mathcal{M}_3\  \ \text{where }\mathcal{M}_3= {\left[\begin{array}{cc}  1/K_{3,+}& 0\\ 0 & 1/K_{3,-} \end{array}\right] \,N_3\, \left[\begin{array}{cc} 0 & K_{3,+}^{^\ell}\\ K_{3,- }^{^\ell }& 0\end{array}\right]}$$

We have $\mathcal{M}_3=(I+h_3\mathcal{R})\mathcal{M}_{0,3}(I+h_3\mathcal{T})  $
with $\mathcal{M}_{0,3}$ obtained by a straightforward calculation as

\bel{monof*}
\mathcal{M}_{0,3}=
\left[\begin{array}{cc} {p}e^{i\phi}\sqrt{1+A^2} &{-{q^{-1}}iA}\\
{i{q}A} & {p^{-1}}e^{-i\phi}\sqrt{1+A^2} \end{array}\right]
\ee
where
\bel{valsA3}
A=e^{\pi\beta/2\hb},\ \ e^{i\phi}=e^{i\phi_2}\left(\hb/2\right)^{i\beta/2\hb},\ \ \ 
\phi_2=\arg\Gamma\left(\frac 12+\frac{i\beta}{2\hb}\right)
\ee
and 
$${p=e^{ -i\frac{\sqrt{E}}{\hb}(C_{-}+C_+)}, \ \ q=e^{- i\frac{\sqrt{E}}{\hb}(C_{-}-C_+)}}=e^{i\frac{\sqrt{E}}{\hb}(C_+-C_-)}$$

Note that the entries $\mathcal{M}_{3,ij}$ of the monodromy matrix $\mathcal{M}_3$ are linked to the entries $\mathcal{M}_{0,3,ij}$ of $\mathcal{M}_{0,3}$ by $\mathcal{M}_{3,ij}=\mathcal{M}_{0,3ij}(1+h_3\mathcal{P}_{ij})$ where $\mathcal{P}_{ij}$ is multilinear in the entries of $\mathcal{R}$, $\mathcal{T}$ and bounded in the parameters.

The entries $\mathcal{S}_{ij}$ of the scattering matrix $\mathcal{S}$ can now be calculated as
$${\mathcal{S}_{11}}=\frac{\det\mathcal{M}_3}{\mathcal{M}_{3;22}}={e^{i\phi -i\frac{\sqrt{E}}{\hb}(C_{-}+C_+)}}\frac{1}{\sqrt{1+A^2}}\,(1+h_3e_{21})$$
and
$${{\mathcal{S}_{12}}}=-\frac{\mathcal{M}_{3;21}}{\mathcal{M}_{3;22}}={-ie^{i\phi{-i\frac{2\sqrt{E}C_-}{\hb}}}}\frac{1}{\sqrt{1+A^2}}\,(1+h_3e_{11}).$$
{(with notations as in \eqref{valsA3}).
Using \eqref{cppcm}, \eqref{valsA3} we obtain}
\bel{argS11n}
\phi -\frac{\sqrt{E}}{\hb}(C_{-}+C_+)=\frac{1}{\hb}(I_+(E)+I_-(E))+\phi_2+\frac{\beta}{2\hb}\left[ 1+\ln(2\hb/|\beta|)\right]
\ee
where $I_+(E),I_-(E)$ are defined by \eqref{defIpbn}. Therefore
\bel{S11n}
\mathcal{S}_{11}=e^{ \frac i{\hb}\int_{-\infty}^{+\infty}\left(\sqrt{E-V(\xi)}-\sqrt{E}\right)d\xi}\, e^{i\theta}\, \frac{1}{\sqrt{1+A^2}}\,\left(1+h_3e_{11}\right)
\ee
where 
\bel{defthetan}
\theta=\phi_2+\frac{\beta}{2\hb}\left[ 1+\ln(2\hb/|\beta|)\right],\ \ \ A=e^{\pi\beta/2\hb}
\ee
(cf. also \eqref{formtheta}).

Also, using \eqref{Ipbn}, \eqref{valsA3} we obtain
\bel{argS21n}
{\phi -\frac{2\sqrt{E}}{\hb}C_-=\frac{2}{\hb}I_-(E)+\phi_2 +\frac{\beta}{2\hb}\left[ 1+\ln(2\hb/|\beta|)\right]+\frac{2}{\hb}\gamma^{-1}\phi_\omega}
\ee
therefore
\bel{S21n}
\mathcal{S}_{21}=-ie^{ \frac{2i}{\hb}\int_{-\infty}^{0}\left(\sqrt{E-V(\xi)}-\sqrt{E}\right)d\xi}\  e^{\frac{i}{\hb}2\gamma^{-1}\phi_\omega}\  e^{i\theta}\,\frac{1}{\sqrt{1+A^2}}\,\left(1+h_3e_{11}\right)
\ee

\

{\em For $h_3$ small} we have, as in \S\ref{smallh1}, that $\theta=O(h_3)$. 

\newpage

\section{Appendix A: Some properties of Gegenbauer polynomials }\label{regatta}

 For reference on this classical topic see e.g.~\cite{Ismail}. 
The Rodrigues formula
$$\tilde{C}^{(\lambda)}_n(t)=(1-t^2)^{-\lambda+1/2}\frac{d^n}{dt^n}\left[ (1-t^2)^{n+\lambda-1/2}\right]$$
define the Gegenbauer (ultraspherical) polynomials ${C}^{(\lambda)}_n(t)$ up to a multiplicative factor: 
${C}^{(\lambda)}_n(t)=K(n,\lambda)\tilde{C}^{(\lambda)}_n(t)$ where
$$K(n,\lambda)=\frac{(-2)^n}{n!}\, \frac{\Gamma(n+\lambda)\,\Gamma(n+2\lambda)}{\Gamma(\lambda)\,\Gamma(2n+2\lambda)}$$
The polynomials ${C}^{(\lambda)}_n(t)$ are of degree $n$ and form a basis in the Hilbert space $L^2([-1,1])$ endowed with the measure $(1-t^2)^{\lambda-1/2}\,dt$.
In particular, ${C}^{(1)}_n(t)=U_n(t)$ are called Cebyshev polynomials of the second kind.

Using the Rodrigues formula it is easy to check that
\bel{eqCU}
(1-t^2)\, \frac{d}{dt}\tilde{C}^{(2)}_{n-1}(t)-3t\tilde{C}^{(2)}_{n-1}(t)=\tilde{U}_n(t)
\ee
and therefore
\bel{eqCUnt}
(1-t^2)\, \frac{d}{dt}{C}^{(2)}_{n-1}(t)-3t{C}^{(2)}_{n-1}(t)=\,-\frac{n(n+2)}{2}\, {U}_n(t)
\ee
or, in integral form,
\bel{cintU}
{C}^{(2)}_{n-1}(t)=\frac{-n(n+2)}{2}\, (1-t^2)^{-3/2}\int_{-1}^t \, {U}_n(\tau)\sqrt{1-\tau^2}\, d\tau
\ee
for all $n\Ge 1$.

Note that all polynomials ${C}^{(\lambda)}_{n}$ are even functions for $n$ even, and odd functions for $n$ odd, as it is easy to see from their recurrence formula:
$$n{C}^{(\lambda)}_{n}(t)=2t(n+\lambda-1){C}^{(\lambda)}_{n-1}(t)-(n+2\lambda-2){C}^{(\lambda)}_{n-2}(t),\ \ \ \ {C}^{(\lambda)}_{-1}=0, {C}^{(\lambda)}_{0}=1$$

\section{Appendix B}\label{WeberFunc}

This section uses notations and results of \cite{Olver1975}, \cite{Olver1959}, \cite{AS}, \cite{nist} to collect and deduce further results on the modified parabolic cylinder functions $E(a,x),\, E^*(a,x)$ and to derive their monodromy matrix. These functions are solutions of the Weber equation: 
\bel{Weber}
 \frac{d^2w}{dx^2}=\left(a-\frac{x^2}{4}\right)w
 \ee

Consider the real-valued, independent solutions $W(a,x),\ W(a,-x)$ of \eqref{Weber} and its complex solutions 
\begin{align*} E(a,x) &=k^{-1/2}W(a,x)+ik^{1/2}W(a,-x) \\
E^*(a,x) &=k^{-1/2}W(a,x)-ik^{1/2}W(a,-x)
\end{align*}
which satisfy
\bel{asyE}
\begin{array}{l}\displaystyle{E(a,x) = 2^{\frac 12}\,e^{i\frac{\pi}{4}+\frac{i}{2}\phi_2}\, x^{-\frac 12-ia}e^{ix^2/4}  }{(1+O(x^{-2}))}\ \ \ \ (x\to +\infty)\\
\displaystyle{E^*(a,x) =  2^{\frac 12}\,e^{-i\frac{\pi}{4}-\frac{i}{2}\phi_2}\, x^{-\frac 12+ia}e^{-ix^2/4}  }{(1+O(x^{-2}))}\ \ \ \ (x\to +\infty)\end{array}
\ee
where $\phi_2=\arg\Gamma(\frac 12+ia)$ and $k=\sqrt{1+e^{2\pi a}}-e^{\pi a}$.

The complex solutions $E(a,x)$, and its complex conjugate $E^*(a,x)$, are entire functions in $x$. They are also real-analytic in the parameter $a$ (this can be seen in the representation \eqref{symbE}, with  \eqref{subwh}, \eqref{defvas}).
If $a>0$ the functions $W(a,\pm x)$ have an oscillatory character for $|x|>2\sqrt{a}$, while between the turning points $x=\pm 2\sqrt{a}$ they have an exponential character. For $a<0$ there is no turning point and these functions are oscillatory on the whole real line.

\subsection{Approximation of $E(a,x)$ with the error behaving like a symbol} {Let $a_0>0$.}

Rewriting \eqref{asyE} as 
\begin{equation}
\label{symbE}
\begin{split}
 E(a,x)&=C(a)\, x^{-\frac 12-ia}e^{ix^2/4}  (1+e_E(x,a)), \\
 E^*(a,x)&=C^*(a)\, x^{-\frac 12-ia}e^{ix^2/4}  (1+e_E^*(x,a))
\end{split}
\end{equation}
the fact that the errors satisfy
\bel{errsymE}
|\partial_x^k e_E(x,a)|\Le C_k \la x\ra^{-2-k},\ \ \ |\partial_x^k e_E^*(x,a)|\Le C_k\la x\ra^{-2-k}\ \ {\text{for }x>0,\ |a|\Le a_0}
\ee
is seen by expressing them as Laplace transforms as follows.

\subsubsection{Estimate of $e_E(\sqrt s,a)$}\label{EstiEax} 

Substituting in \eqref{Weber} 
\bel{subwh}
w(x)=e^{ix^2/4} h(a,x^2),\ \ x=\sqrt{s}
\ee
we can calculate the Laplace representation  
\bel{defvas}
h(a,s)=\int_0^{+\infty}e^{-ps}(1+2ip)^{-\frac 34 -\frac{ia}{2}}p^{-\frac 34+\frac{ia}{2}}\, dp
\ee
In \eqref{defvas} we use the Taylor expansion with remainder:
$$(1+2ip)^{{-\frac 34 -\frac{ia}{2}}}=1+R(p)\, p$$ where $$R(p)=({-\frac 34 -\frac{ia}{2}})(1+2i\xi_p)^{{-\frac 74 -\frac{ia}{2}}}$$ (for some $\xi_p\in(0,p)$)
and obtain$$h(a,s)=\frac{\Gamma(\frac14+\frac{ia}2)}{s^{\frac14+\frac{ia}2}} + \int_0^{+\infty}e^{-ps}R(p) p^{\frac 14+\frac{ia}{2}}\, dp $$
which compared with \eqref{symbE} gives
$$h(a,s)=\frac{\Gamma(\frac14+\frac{ia}2)}{s^{\frac14+\frac{ia}2}}\, [1+e_E(\sqrt s,a)]$$
{where }
\bel{linkeh}
e_E(\sqrt s,a)=\frac{s^{\frac14+\frac{ia}2}}{\Gamma(\frac14+\frac{ia}2)} \int_0^{+\infty}e^{-ps}R(p) p^{\frac 14+\frac{ia}{2}}\, dp
\ee
Using the estimate $$|(1+2i\xi_p)^{{-\frac 74 -\frac{ia}{2}}}|=|1+2i\xi_p)|^{-\frac 74} \exp[\frac a2 \arg(1+2i\xi_p)]<e^{\pi a/2}$$ we obtain
$$
|e_E(\sqrt s,a)|  \Le   \frac{s^{\frac14} \, e^{\pi a/2}|-\frac 34 -\frac{ia}{2}|}{\big| \Gamma(\frac14+\frac{ia}2)\big|} \frac{\Gamma(5/4)}{s^{5/4}}\Le \rm{Const.}\ \frac 1s\ \ \text{for all }|a|<a_0$$

\subsubsection{Estimate of derivatives}\label{DexEax} 
Taking the derivative in $s$ in \eqref{linkeh} we obtain
$$\partial_s e_E(\sqrt s,a)=\frac 1s (\frac 14+\frac {ia}2) e_E(\sqrt s,a)- \frac{s^{\frac14+\frac{ia}2}}{\Gamma(\frac14+\frac{ia}2)} \int_0^{+\infty}e^{-ps}R(p) p^{\frac 54+\frac{ia}{2}}\, dp
$$
and the same method as in \S\ref{EstiEax} proves that $|\partial_s e_E(\sqrt s,a)|\Le$ Const. $s^{-2}$ for all $a$ with $|a|<a_0$.

In the same way, by induction, it can be shown that higher order derivatives satisfy $|\partial_s^k e_E(\sqrt s,a)|<C_ks^{-1-k}$.
We also note that $e_E(\sqrt s)$ depends analytically on $a$ for $a\in(-a_0,a_0)$.

\subsubsection{Final remark} The proof of \eqref{errsymE} is completed by noting that:

\begin{Remark} If a function behaves like a symbol in variable $s$ then it also behave like a symbol in variable $x=s^{\alpha}$.
\end{Remark}
Indeed, if $|\partial_s^k F(s)|\lesssim \la s\ra^{c-k}$ then, since $\partial_s\sim x^{1-1/\alpha}\partial_x$ then $$|\partial_x F|\lesssim x^{-k+k/\alpha}\la x^{1/\alpha}\ra ^{c-k}\sim \la x \ra^{c/\alpha-k}$$

\subsection{Some bounds on the function $E(a,x)$}\label{remonE}

(i) We have
\bel{sqxE}
\sqrt{x}|E(a,x)|<M_0\ \ \ \text{for all }x\Ge 1,\ \ |a|\Le a_0
\ee

 Indeed, from \eqref{symbE}, \eqref{errsymE} we see that  there is some $x_1$ (large enough, depending only on $a_0$, but not on $a$) so that \eqref{sqxE} holds with $M_0=2$ (or any $M_0>2^{1/2}$) for all $x>x_1$ and $|a|\Le a_0$ . Also, let $M_1$ be the maximum of $\sqrt{x}|E(a,x)|$ for $x\in[1,x_1]$ and $|a|\Le a_0$. Then let $M_0=\max\{2,M_1\}$.

(ii) $E(a,x)\ne 0$ for all $x$ and $a$.

In fact, the modulus $F=|E(a,x)|$ satisfies the differential equation $$F''-F^{-3}+(x^2/4-a)F=0$$ \cite{Miller} therefore $F$ has no zeroes.

(iii) $|E(a,x)|\Ge C>0$ for all $x\in[0,1]$ and $|a|\Le a_0$.

To see this,  let $C$ be the minimum of the continuous function $F$.

\subsubsection{Further results}

 We need the following estimate:
\bl\label{DaG0}
The function $G_0(a,x)=E'(a,x)/E(a,x)$ satisfies
$$|\partial_a^\ell G_0|\Le C_\ell\la x\ra$$
for all $a$ in a compact set.
\el

\begin{proof}
The result is obtained by expressing $G_0$ as a Laplace transforms as follows. 
$G_0$ satisfies the differential equation $$G_0'+G_0^2=a-x^2/4$$ and, by \eqref{asyE}, $G_0(x)\sim ix/2$ for $x\to\infty$. Denoting $$G_0(x)=x(\frac{i}{2}+u(x^2)),\quad x^2=s,\quad u(s)=(\mathcal{L}U)(s)=\int_0^\infty e^{sp}U(a,p)\, dp$$ then $U(a,p)$ satisfies the integral equation
$$\left( \frac i2-p \right) U  +\frac i4-\frac a2+\frac i2\,\int_{0}^{p}\!U  \left(a, q \right) \,{\rm d}q+U*U=0$$
which has a solution satisfying $|U(a,p)|<\exp(-\nu p)$ for some $\nu>0$ \cite{Duke} and which can be chosen independent of $a$
for $a$ in a compact set, say $|a|\Le a_0$. Since $\partial_a^\ell u=\mathcal{L}(\partial_a^\ell U)$ and $U$ is entire in $a$, Cauchy's integral formula (in $a$) shows that $|\partial_a^\ell U|<C_\ell \exp(-\nu p)$ and therefore, for $s\Ge s_0$, $$(\mathcal{L}\partial_a^\ell U)(s)<C_\ell \sup |U|\int_0^\infty e^{-(s-\nu)p}\, dp \lesssim (s-\nu)^{-1}$$
Therefore $\partial_a^\ell G_0$ has the same decay in $x$ as $G_0$. \end{proof}

\subsection{Monodromy of the modified parabolic cylinder functions} 

If $w(x)$ solves \eqref{Weber} then so does $w(-x)$, hence $E(a,-x)$ and $E^*(a,-x)$ are also solutions and \eqref{asyE} (with $x$ replaced by $-x$) gives their asymptotics for $x\to -\infty$.

Using the connection formula (\cite{AS}\S 19.18.3)
$$\sqrt{1+e^{2\pi a}}\, E(a,x)-e^{\pi a}E^*(a,x)=iE^*(a,-x)$$
implying also
$$\sqrt{1+e^{2\pi a}}\, E(a,-x)-e^{\pi a}E^*(a,-x)=iE^*(a,x)$$
we obtain the monodromy matrix: $\left[E(a,-x)\ E^*(a,-x)\right]=\left[E(a,x)\ E^*(a,x)\right]\, M_E$ where

\bel{monoE}
M_E=\left[\begin{array}{cc} -ie^{\pi a} &-i\sqrt{1+e^{2\pi a}}\\
i\sqrt{1+e^{2\pi a}} & ie^{\pi a} \end{array}\right]
\ee

\

The form of \eqref{Weber} used in \S\ref{Aicase} is \eqref{eqw}
(linked to \eqref{Weber} by changing  $x=y_1\frac{\sqrt{2}}{ \sqrt{h_1}}=y\frac{\sqrt{2}}{ \sqrt{\hb}}$, $ a=\frac{1}{2h_1}=\frac{\beta}{2\hb}$, see \eqref{psi12}) for which we have the fundamental system
\bel{fdsysE}
E\left({\frac{\beta}{2\hb}},y\frac{\sqrt{2}}{ {\sqrt{\hb}}}\right):=\mu\phi_+(y),\ \ E^*\left({\frac{\beta}{2\hb}},y\frac{\sqrt{2}}{ {\sqrt{\hb}}}\right):=\overline{\mu}\phi_-(y)
\ee
where
\bel{defphiE}
\begin{array}{l} \phi_\pm(y)= y^{-\frac 12\mp i\frac{\beta}{2\hb}}\, e^{\pm iy^2/(2\hb)}(1+o(1))\ \ \ (y\to \infty) \\ \ \\
\mu=2^{\frac14 -i \frac{\beta}{4\hb}}\, \hb^{\frac14 +i \frac{\beta}{4\hb}}\, e^{i\frac{\pi}{4}+\frac{i}{2}\phi_2},\ \ \ \phi_1=\arg\Gamma\left(\frac 12+i\frac{\beta}{2\hb}\right)
\end{array}
\ee

With the notation
$$E\left({\frac{\beta}{2\hb}},-y\frac{\sqrt{2}}{ {\sqrt{\hb}}}\right):=\mu\phi_+^\ell(y),\ \ E^*\left({\frac{\beta}{2\hb}},-y\frac{\sqrt{2}}{ {\sqrt{\hb}}}\right):=\overline{\mu}\phi_-^\ell(y)$$
\eqref{monoE} yields $\left[\phi_+^\ell\ \  \phi_-^\ell\right]=\left[\phi_+\ \ \phi_-\right] \tilde{M}$ where

\bel{monoP}
\tilde{M}=\left[\begin{array}{cc} -ie^{\pi a} &\theta^{-1}\sqrt{1+e^{2\pi a}}\\
\theta\sqrt{1+e^{2\pi a}} & ie^{\pi a} \end{array}\right],\ \ \ {\rm with\ }a=\frac{\beta}{2\hb},\ \theta=(2/\hb)^{ \frac{i\beta}{2\hb}}e^{-i\phi_2}
\ee

\

{\em The monodromy for the solutions $w_\pm$ in \S\ref{Aicase}.}

The solutions $w_\pm$ given by Lemma\,\ref{L5}  are {linked} to $\phi_\pm$ of \eqref{defphiE} by Lemma\,\ref{WtoAB} and it follows that 
\bel{monoW}
[w_+^\ell\ w_-^\ell]=[w_+\ w_-]M
\ee
where
\bel{monoWM}
{M}=\left[\begin{array}{cc} -iA &i\theta_1\sqrt{1+A^2}\\
-i\theta_1^{-1}\sqrt{1+A^2} & iA \end{array}\right]\ \ \ \ {\rm where\ }A=e^{\pi\beta/(2\hb)},\ \theta_1=e^{i\phi_2}(2e\hb/\beta)^{i\beta/(2\hb)}
\ee

\vskip3cm

{\em The monodromy for the solutions $w_{3,\pm}$ in \S\ref{Bcase}.}

The form \eqref{eqw3} of Weber's equation used in \S\ref{Bcase} is linked to  \eqref{Weber} by changing  $x=y_3\frac{\sqrt{2}}{ \sqrt{h_3}}=y\frac{\sqrt{2}}{ \sqrt{\hb}}$, $ a=\frac{-1}{2h_3}$. Noting that $1/(2h_1)=-1/(2h_3)=\beta/(2\hb)$
equation  \eqref{eqw3} admits the same fundamental system \eqref{fdsysE}. 

The solutions $w_{3,\pm}$ given by Corollary~\ref{C63}  are {linked} to $\phi_\pm$ of \eqref{defphiE} by Lemma~\ref{WtoAB3}, yielding that
$$w_{3,\pm}=\lam_{3,\mp}\, |\beta|^{\frac 14\pm i\beta/4\hb} \phi_\pm$$
and it follows that
\bel{monoW3}
[w_{3,+}^\ell\ w_{3,-}^\ell]=[w_{3,+}\ w_{3,-}]M_3
\ee
where
\bel{monoWM3}
{M_3}=\left[\begin{array}{cc} -iA &i\theta_3\sqrt{1+A^2}\\
-i\theta_3^{-1}\sqrt{1+A^2} & iA \end{array}\right] 
\ee
with
\[ A=e^{\pi\beta/(2\hb)},\ \theta_3=e^{i\phi_2}(2e\hb/|\beta|)^{i\beta/2\hb}e^{2iC_3\beta/\hb}
\]

\section{Appendix C}\label{PropC8}

We collect here some results found in \cite{CDST} that we use, referring to symbol behavior of solutions of Volterra equations.

The following is Proposition C8 in \cite{CDST}.
\bp
 \label{prop:volterraAinox}
Fix $x_0 \in \mathbb{R}$, $\lambda_0>0$, $\alpha >-\frac12$, $\beta\geq \frac32 \gamma \geq 0$ and assume that $\beta-(\alpha+\frac12)\gamma\geq 0$. 
Let $c$ be a real--valued function that satisfies $c(\lambda)\geq x_0$ and $|c^{(\ell)}(\lambda)|\leq C_\ell \lambda^{-\gamma-\ell}$ for all $\lambda \in (0,\lambda_0)$, $\ell \in \mathbb{N}_0$.
Furthermore, assume that the (possibly complex--valued) functions $a(\cdot,\lambda)$, $b(\cdot,\lambda)$ satisfy the bounds
$$ |\partial_\lambda^\ell \partial_x^k a(x,\lambda)|\leq C_{k,\ell}\langle x \rangle^{-k}\lambda^{-\ell},\quad
|\partial_\lambda^\ell \partial_x^k b(x,\lambda)|\leq C_{k,\ell}\langle x \rangle^{\alpha-k}\lambda^{\beta-\ell} $$
for all $x_0\leq x \leq c(\lambda)$, $\lambda \in (0,\lambda_0)$ and $k,\ell \in \mathbb{N}_0$.
Set
$$ K(x,y,\lambda):=Bi(y)^2 b(y,\lambda)\int_x^y Bi(u)^{-2}a(u,\lambda)du $$
for $x_0\leq y \leq x \leq c(\lambda)$. 
Then the equation
$$ \varphi(x,\lambda)=\int_{x_0}^x K(x,y,\lambda)[1+\varphi(y,\lambda)]dy $$
has a unique solution $\varphi(\cdot,\lambda)$ that satisfies
$$ |\partial_\lambda^\ell \partial_x^k \varphi(x,\lambda)|\leq C_{k,\ell}\langle x \rangle^{\alpha+\frac12-k}\lambda^{\beta-\ell} $$
for all $x_0 \leq x \leq c(\lambda)$, $\lambda \in (0,\lambda_0)$ and $k,\ell \in \mathbb{N}_0$.
\ep

The following is Lemma D.5 in \cite{CDST}.

Consider equation (D.9) in  \cite{CDST}, namely
\begin{equation} 
\label{eq:Airypert}
\underbrace{\phi''(\zeta)=\nu^2 \zeta \phi(\zeta)}_{\mbox{Airy}}+\underbrace{V_2(\zeta)\phi(\zeta)}_{\mbox{pert.}} 
\end{equation}
where $V_2$ satisfies the bounds 
\bel{boundsV2}
|V_2^{(k)}(\zeta)|\leq C_k \langle \zeta \rangle^{-2-k}
\ee
for all $k \in \mathbb{N}_0$ and all $\zeta \in \mathbb{R}$.

\begin{Lemma}
 \label{lem:bessel-}
 For $\zeta \leq 0$, $\nu \geq 1$ there exists a fundamental system $\{\phi_\pm(\cdot, \nu)\}$ of Eq.~\eqref{eq:Airypert} of the form
 $$ \phi_\pm(\zeta, \nu)=[Ai(\nu^\frac23 \zeta)\pm iBi(\nu^\frac23 \zeta)][1+\nu^{-1}a_\pm(\zeta, \nu)] $$
 where the functions $a_\pm(\cdot, \nu)$ are smooth and $|a_\pm(\zeta,\nu)|\lesssim 1$ in the above range 
 of $\zeta$ and $\nu$.
 Furthermore, $a_\pm$ satisfy the bounds
 \bel{estimcRD}
  |\partial_\nu^\ell \partial_\zeta^k a_\pm(\zeta, \nu)|\leq C_{k,\ell}\langle \zeta\rangle^{-\frac32-k}\nu^{-\ell},\quad \zeta \leq -1 
  \ee
 as well as
 \bel{estimcRD1} |\partial_\nu^\ell a_\pm(0,\nu)|\leq C_\ell \nu^{-\ell},\quad 
 |\partial_\nu^\ell \partial_\zeta a_\pm(0,\nu)|\leq C_{k,\ell}\nu^{\frac23-\ell}
 \ee
 for all $\nu \geq 1$ and $k,\ell \in \mathbb{N}_0$.
\end{Lemma}

\end{document}